\documentclass[11pt]{article}
	
	\newcommand{\blind}{0}
	
    \makeatletter
    \renewcommand\section{\@startsection {section}{1}{\z@}%
                                       {-3.5ex \@plus -1ex \@minus -.2ex}%
                                       {2.3ex \@plus.2ex}%
                                       {\normalfont\fontfamily{phv}\fontsize{14}{17}\bfseries}}
    \renewcommand\subsection{\@startsection{subsection}{2}{\z@}%
                                         {-3.25ex\@plus -1ex \@minus -.2ex}%
                                         {1.5ex \@plus .2ex}%
                                         {\normalfont\fontfamily{phv}\fontsize{12}{14}\bfseries}}
    \renewcommand\subsubsection{\@startsection{subsubsection}{3}{\z@}%
                                        {-3.25ex\@plus -1ex \@minus -.2ex}%
                                         {1.5ex \@plus .2ex}%
                                         {\normalfont\normalsize\fontfamily{phv}\fontsize{11}{13}\selectfont}}
    \makeatother
	
	\usepackage{xcolor}
	\usepackage{url} 
	\usepackage{amsthm,amsmath,amsfonts,amssymb,amsbsy,nccmath}
    \usepackage[authoryear]{natbib}
    \usepackage[colorlinks,citecolor=blue,urlcolor=blue]{hyperref}
    \usepackage{algorithm2e}
    \usepackage{graphicx,psfrag,epsf}
    \usepackage{enumerate}
    \usepackage{placeins}
        \usepackage{comment}
    \graphicspath{ {./images/} }
    \usepackage{tabularx,booktabs}
    \usepackage{placeins}
    \usepackage{diagbox}
    \usepackage{multirow}
    \usepackage{adjustbox}
    \usepackage[createShortEnv]{proof-at-the-end}

    \newcommand\clearrow{\global\let\rowmac\relax}
    \newcolumntype{Y}{>{\centering\arraybackslash}X}
    
    \theoremstyle{plain}

    \newtheorem{theorem}{Theorem}
    
    \theoremstyle{remark}
    \newtheorem{definition}{Definition}

    
    \DeclareMathOperator*{\argmax}{arg\,max}  
    \DeclareMathOperator*{\argmin}{arg\,min}  

\begin{document}
		
		\def\spacingset#1{\renewcommand{\baselinestretch}%
			{#1}\small\normalsize} \spacingset{1}
		
		\if0\blind
		{
			\title{\Large \bf \emph{Bayesian Clustering Prior with Overlapping Indices for Effective Use of Multisource External Data}}
			\author{Xuetao Lu $^a$ and J. Jack Lee $^{a*}$ \\
			$^a$ Department of Biostatistics, The University of Texas MD Anderson Cancer Center \\
			$^*$ Author for correspondence: jjlee@mdanderson.org}
			\date{}
			\maketitle
		} \fi
		
		\if1\blind
		{

            \title{\bf \emph{IISE Transactions} \LaTeX \ Template}
			\author{Author information is purposely removed for double-blind review}
			
\bigskip
			\bigskip
			\bigskip
			\begin{center}
				{\LARGE\bf \emph{IISE Transactions} \LaTeX \ Template}
			\end{center}
			\medskip
		} \fi
		\bigskip
		
	\begin{abstract}
The use of external data in clinical trials offers numerous advantages, such as reducing the number of patients, increasing study power, and shortening trial durations. In Bayesian inference, information in external data can be transferred into an informative prior for future borrowing (i.e., prior synthesis). However, multisource external data often exhibits heterogeneity, which can lead to information distortion during the prior synthesis. Clustering helps identifying the heterogeneity, enhancing the congruence between synthesized prior and external data, thereby preventing information distortion. Obtaining optimal clustering is challenging due to the trade-off between congruence with external data and robustness to future data. We introduce two overlapping indices: the overlapping clustering index (OCI) and the overlapping evidence index (OEI). Using these indices alongside a K-Means algorithm, the optimal clustering of external data can be identified by balancing the trade-off. Based on the clustering result, we propose a prior synthesis framework to effectively borrow information from multisource external data. By incorporating the (robust) meta-analytic predictive prior into this framework, we develop (robust) Bayesian clustering MAP priors. Simulation studies and real-data analysis demonstrate their superiority over commonly used priors in the presence of heterogeneity. Since the Bayesian clustering priors are constructed without needing data from the prospective study to be conducted, they can be applied to both study design and data analysis in clinical trials or experiments.
	\end{abstract}
			
	\noindent%

	\spacingset{1.5} 

\section{Introduction}
\label{sec:intro}

Incorporating external data from multiple sources into the design and analysis of clinical trials or experiments is an area of intense research interest. The US Food and Drug Administration (FDA) has released guidance, Use of Real-world Evidence to Support Regulatory Decision-making for Medical Devices, to encourage the use of external data in new studies \citep{FDA2017}. This practice aims to enhance the efficiency of new trials by leveraging real-world data (RWD) to inform study parameters, potentially reducing sample size, increasing the power and precision of study outcomes, and accelerating trial timelines \citep{Spiegelhalter2004}. In a Bayesian framework, the cornerstone of RWD incorporation involves constructing informative priors from external data, a process referred to as evidence synthesis. Methods for synthesizing informative priors are well-studied. For example, the meta-analytic predictive (MAP) prior \citep{Neuenschwander2010,Schmidli2014} uses meta-analysis to summarize information from external data into an informative prior. The power prior (PP) \citep{CHEN2000PP,Ibrahim2003} adjusts the influence of external data on the analysis of current data based on its relevance and reliability, using a likelihood discounting approach. Commensurate priors \citep{Hobbs2011,Hobbs2012} and multisource exchangeability models (MEM) \citep{Kaizer2017} determine the degree of information borrowing from external data according to their relevance and consistency with the new data. The elastic prior \citep{Jiang2023} dynamically borrows information from external data through a monotonic function of a congruence measure between external data and new trial data.

The use of prior information is critical in both trial design and data analysis. During the trial design stage, where new trial data is not yet available, prior information is typically derived from domain knowledge or external data. In the data analysis stage, once new trial data is available, priors can be refined by evaluating the similarity between the external and new trial data. Constructing an informative prior that is suitable for both trial design and data analysis is challenging. For example, with the exception of the MAP prior, all of the aforementioned priors require the new trial data, which limits their applicability during the trial design stage.

Another challenge arises from the diversity and heterogeneity of multiple data sources, including differences in, study design, population characteristics, eligibility criteria and outcome measures \citep{Schmidli2014}. Such variability complicates the accurate transferring or borrowing of information from external data into a prior, potentially leading to information distortion and adversely affecting the analysis of the new trial. 
Therefore, it is crucial to accurately identify the heterogeneity across external datasets. MAP (or rMAP) prior accommodates this heterogeneity by using the linear mixed model with a random effect parameter \citep{RBEST2021}. However, a single parameter may not adequately capture complex heterogeneity structures, such as scenarios where external datasets include multiple clusters with varying degrees of homogeneity. The MEM approach attempts to address this challenge by measuring pairwise exchangeability among parameters estimated from external datasets. While it does not explicitly construct a prior that adapts to the heterogeneity structure. In this study, we focus on clustering methods, where heterogeneity is identified by partitioning parameters (associated with external datasets) into distinct clusters. A widely used approach is the Bayesian nonparametric clustering based on Dirichlet processes \citep{GERSHMAN20121}. For example, \cite{Chen2020} proposed a Bayesian clustering hierarchical model to dynamically partition sub-trials into clusters for efficient information borrowing in basket trials. However, in this method, the number of clusters is determined by a hyperparameter, making it challenging to establish an interpretable and unified criterion for selecting the optimal value of the hyperparameter across different applications \citep{lu2023overlapping}. Additionally, based on our knowledge, this approach has not been applied to prior synthesis with multisource external data. 

To address these challenges, we propose a novel approach based on overlapping coefficients \citep{Weitzman1970,Schmid2006}. Specifically, we introduce the overlapping clustering index (OCI) and a K-Means algorithm to identify the heterogeneity across external datasets. To measure the congruence between the synthesized prior and external data, we define the Overlapping Evidence Index (OEI). A higher OEI signifies more accurate information transfer from external data to the informative prior, reflecting stronger congruence with the external datasets.
However, there is an inherent trade-off between maximizing OEI and maintaining the robustness of the prior. By balancing this trade-off, we provide a criterion for identifying the optimal clustering of external datasets. Compared to aforementioned clustering hyperparameter in the nonparametric Bayesian methods, our criterion is more interpretable and maintains a unified meaning across different applications. Then, using the optimal clustering results, we introduce the (Robust) Bayesian Clustering Prior, which is constructed as a weighted sum of priors synthesized from individual clusters. This framework is highly flexible and can be integrated with other methods. For instance, we propose the Bayesian Clustering MAP (BCMAP) Prior and the Robust Bayesian Clustering MAP (rBCMAP) Prior by incorporating MAP and robust MAP priors, which possess desirable properties and can be used in both trial design and data analysis. 

Section \ref{sec:sec2} introduces the notation and assumptions of Bayesian evidence synthesis from multisource external data. With an example, we illustrate the information distortion due to incorrect heterogeneity identification and discuss the trade-off between evidence congruence and robustness. In Section \ref{sec:sec3}, we propose OCI and OEI to address the challenges of heterogeneity identification and the trade-off between evidence congruence and robustness. The construction of the Bayesian clustering prior, along with an example, is presented in Section \ref{sec:sec4}. We conduct simulation studies to compare our methods with existing methods in Section \ref{sec:sim}. Section \ref{sec:real_data} presents the application to a real dataset. A brief discussion is provided in Section \ref{sec:Con}.

\section{Bayesian Evidence Synthesis from Multisource External data}
\label{sec:sec2}

Let $Y_1,\ldots,Y_H$ denote external data from multiple sources. Assume $\theta$ to be the common parameter of interest. Bayesian evidence synthesis aims to create an informative prior of $\theta$ from $Y_1,\ldots,Y_H$, denoted as $\pi(\theta|Y_1,\ldots,Y_H)$. Then, for any new data $Y^*$, the inference of $\theta$ can borrow the information from $\pi(\theta|Y_1,\ldots,Y_H)$ through $p(\theta|Y^*) \propto L(Y^*|\theta)\pi(\theta|Y_1,\ldots,Y_H)$, where $L(Y^*|\theta)$ is the likelihood. 
In this study, we assume $Y^*$ is unknown; in other words, $Y^*$ has no effect on the construction of $\pi(\theta|Y_1,\ldots,Y_H)$. In Bayesian inference, the prior distribution represents the beliefs or information about a parameter before observing the new data. Avoiding ``use the data twice'' \citep{Carlin2000EmpiricalBP} is fundamental for maintaining the integrity of the Bayesian updating process. Compared to the most existing methods, such as the MEM and power prior, which use $Y^*$ first in the prior and then in the likelihood, this assumption strictly adheres to the isolation principle. It makes $\pi(\theta|Y_1,\ldots,Y_H)$ applicable in both trial design (without $Y^*$) and data analysis (with $Y^*$). Another assumption in this study is the exclusion of covariate information. This assumption stems from the practical challenges of obtaining such information. For instance, patient-level data may be restricted to public, even for research purposes. Even when this information is accessible, the available covariates often vary across data sources. For example, data source 1 might include covariates ${X_1, X_2, X_3}$, data source 2 might include ${X_2, X_3, X_5}$, and data source 3 might include ${X_3, X_4, X_5}$, leaving only $X_3$ as a common covariate. In such cases, inference relying solely on $X_3$, such as covariate adjustment, may lead to questionable conclusions.

\begin{figure}[t!]
\centering
\includegraphics[width=.9\textwidth]{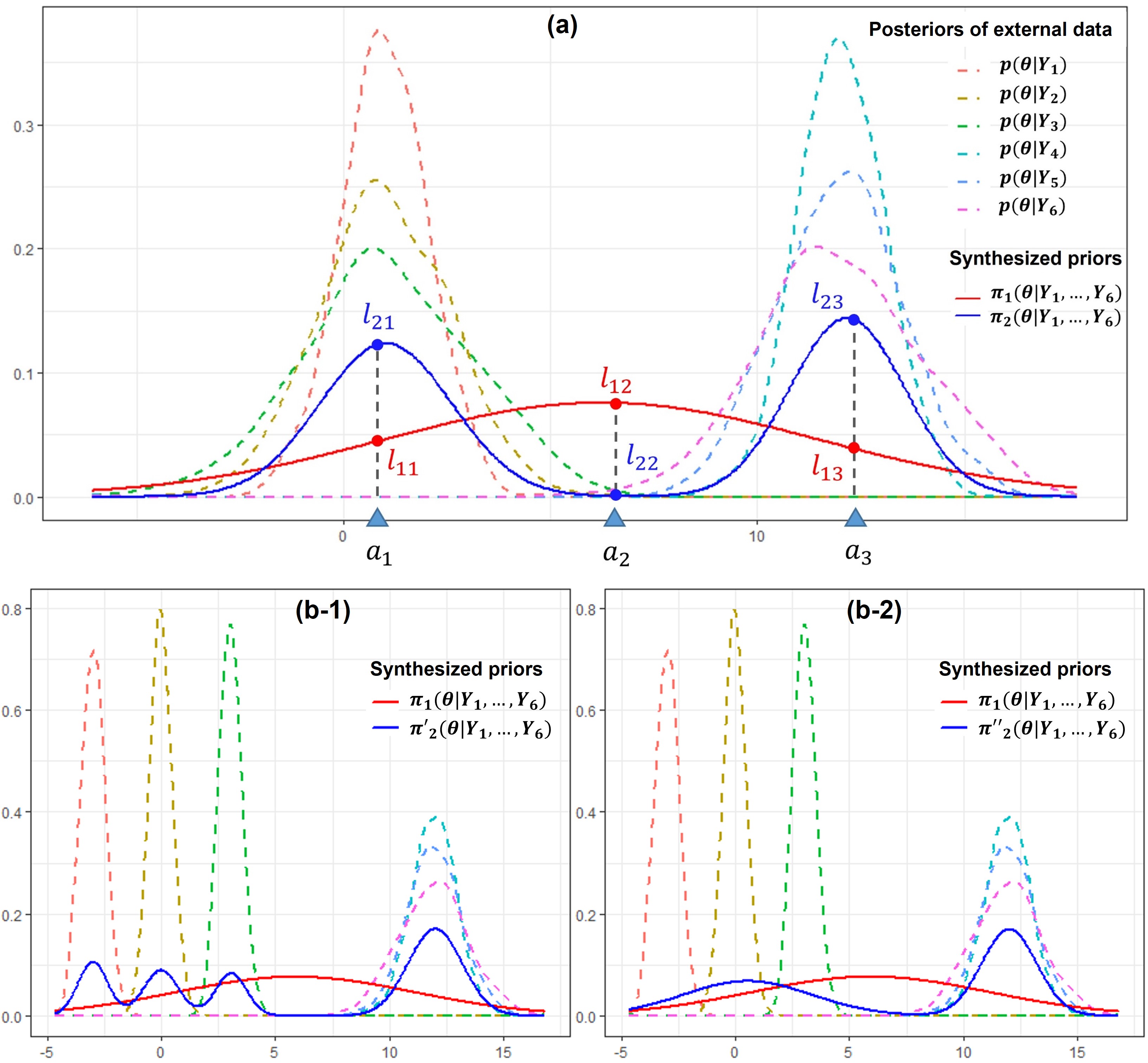}
\caption[challenge1]{Information distortion and the trade-off between evidence congruence and robustness. Panel (a) shows an example illustrating the information distortion caused by heterogeneity with prior $\pi_2$ more appropriate than prior $\pi_1$. In Panel (b-1), the prior $\pi'_2$ has stronger evidence congruence but weaker robustness. In Panel (b-2), the prior $\pi''_2$ has weaker evidence congruence but stronger robustness. Panels (b-1) and (b-2) demonstrate the trade-off between evidence congruence and robustness.}
\label{fig:challenge}
\end{figure}

In practice, heterogeneity often exists among data sources, such as multi-regional data or data from multiple health centers. One approach to address this issue is to filter out certain datasets to achieve homogeneity. However, without knowledge of covariate information or the new data $Y^*$, this method risks losing valuable information, as it is hard to determine which data sources should be excluded. An alternative approach is to include all external datasets. In this scenario, the evidence synthesis model used to construct the prior must be capable of effectively identifying the heterogeneity across the various data sources. Otherwise, it may distort the information transferred from $Y_1,\ldots,Y_H$ to $\pi(\theta|Y_1,\ldots,Y_H)$. The information of $\theta$ in $Y_1,\ldots,Y_H$ is contained in $p(\theta|Y_1),\ldots,p(\theta|Y_H)$ where $p(\theta|Y_h)$, $h=1,\ldots,H$, is obtained through $p(\theta|Y_h)\propto L(Y_h|\theta)\pi(\theta)$, and $\pi(\theta)$ can be either a weakly informative or an informative prior of $\theta$. (Note: Instead of using ${\theta_1, \ldots, \theta_H}$, we utilize the posteriors ${p(\theta | Y_1), \ldots, p(\theta | Y_H)}$ to reflect the heterogeneity of the external data. This notation emphasizes the consistency between the external data and the synthesized prior $\pi(\theta | Y_1, \ldots, Y_H)$.) The information distortion can be illustrated by the example shown in panel (a) of Figure \ref{fig:challenge}. The example includes six external datasets with corresponding posteriors $p(\theta|Y_1),\ldots,p(\theta|Y_6)$. It is obvious that these datasets are heterogeneous and can be partitioned into two clusters. Consider the information transferring at three points, $\theta=a_1$, $\theta=a_2$, and $\theta=a_3$. We use the likelihood of the priors to measure the transferred information. The corresponding likelihoods are $l_{11},l_{12},l_{13}$ for $\pi_1$ and $l_{21},l_{22},l_{23}$ for $\pi_2$. According to the posteriors $p(\theta|Y_1),\ldots,p(\theta|Y_6)$, it is clear that the likelihoods at $\theta=a_1$ and $\theta=a_3$ should be greater than the likelihood at $\theta=a_2$. However, $\pi_1$ performs oppositely, with $l_{11},l_{13}<l_{12}$. This occurs because the synthesizing method of $\pi_1$ incorrectly identifies the heterogeneity, thus distorting the information in the external data. Conversely, $\pi_2$ correctly pinpoints the heterogeneity and accurately reflects the information in the external data, with $l_{21},l_{23}<l_{22}$.

As discussed above, the quality of a synthesized prior can be evaluated by the consistency of information about $\theta$ in $Y_1,\ldots,Y_H$ and in $\pi(\theta|Y_1,\ldots,Y_H)$. We refer to this consistency as the evidence congruence of $\pi(\theta|Y_1,\ldots,Y_H)$ with respect to $Y_1,\ldots,Y_H$. In panel (a) of Figure \ref{fig:challenge}, it is clear that $\pi_2$ has stronger evidence congruence than $\pi_1$. However, higher evidence congruence is not always better. Robustness is another crucial criterion for evaluating the quality of a synthesized prior \citep{Schmidli2014}. It is easy to check that the evidence congruence of $\pi'_2$ in Panel (b-1) is greater than that of $\pi''_2$ in Panel (b-2). But we prefer $\pi''_2$ because it is more robust. $\pi'_2$ is constructed by identifying four clusters in the external datasets, whereas $\pi''_2$ assumes two clusters. In sum, both criteria of evidence congruence and robustness are closely related to the heterogeneity identification. Since $p(\theta|Y_1),\ldots,p(\theta|Y_H)$ contain all the information about $\theta$, an accurate clustering of $p(\theta|Y_1),\ldots,p(\theta|Y_H)$ can strike a good balance between evidence congruence and robustness, thereby helping to create a high-quality informative prior.

\section{Overlapping Indices}\label{sec:sec3}

Overlapping coefficient (OVL) is a measure of the intersection area between two probability density or mass functions. Let $X$ and $Y$ be two random variables with probability density or mass functions $f$ and $g$, respectively. $\Omega$ is the common support of $f$ and $g$. OVL can be defined in equation \eqref{eq:ovl-1}, where the integral expression can be used in this paper without loss of generality. 
\begin{equation}\label{eq:ovl-1}
     OVL(X,Y) =
     \begin{cases}
        \int_\Omega min\{f(t),g(t)\}\,dt, & \text{continuous}\\
        \sum_\Omega min\{f(t_i),g(t_i)\}I(t_i \in \Omega), & \text{discrete}
    \end{cases}
\end{equation}
Based on the concept of OVL, we propose two overlapping indices for the clustering of $p(\theta|Y_1),\ldots,p(\theta|Y_H)$ to address the challenges discussed in Section \ref{sec:sec2}.

\subsection{Overlapping Clustering Index and K-Means Clustering}\label{sec:ociKM}

In our prior work, \cite{lu2023overlapping}, we proposed the overlapping clustering index (OCI) and a K-Means clustering algorithm to cluster $n$ heterogeneous random variables into $K$ clusters ($n\geq K$). We have slightly modified that definition to fit our requirements here:

\begin{definition}[Overlapping Clustering Index]\label{oci}
Let $p(\theta|Y_h)$, $h=1,\ldots,H$, be the probability density function (pdf) or probability mass function (pmf) of posterior distributions obtained from the external datasets $Y_1,\ldots,Y_H$. The corresponding random variables are denoted as $\theta|Y_h$, $h=1,\ldots,H$. A map $S$ partitions them into $K$ clusters $\{G_1,...,G_K\}$, $1 \le K \le H$. For each cluster $G_m$, $m=1,...,K$, let $Z_m$ be a Gaussian random variable with density function $g_m$, which is the maximum likelihood estimation (MLE) obtained from the samples of the random variables in cluster $G_m$. Then, the OCI of this partition $K$ is defined as follows:
\begin{equation}\label{eq:oci}
     OCI_K=\sum^K_{m=1}\sum^H_{h=1}OVL(\theta|Y_h,Z_m)\cdot I(\theta|Y_h\in G_m),
\end{equation} 
where $I(\cdot)$ is a indicator function. (Note: The random variables $\theta|Y_h$ in the same cluster $G_m$ share the same mean (the cluster mean) and have finite variances. By the central limit theorem, when the number of $\theta|Y_h$ in cluster $G_m$ goes to infinity, their summation approximates a Gaussian distribution. Given the additive property of the expression $\sum^H_{h=1}OVL(\theta|Y_h,Z_m)\cdot I(\theta|Y_h\in G_m)$, we choose a Gaussian random variable $Z_m$.)
\end{definition}
The $OCI_K$ measures the overall within cluster homogeneity of the K-partition. Based on it, we can define the optimal $S_K^{(oci)*}$ for K partition as follows:
\begin{equation}\label{eq:ocis*}
S_K^{(oci)*} = \argmax_{S_K} OCI_K = \argmax_{S_K} \sum^K_{m=1}\sum^H_{h=1}OVL(\theta|Y_h,Z_m)\cdot I(\theta|Y_h\in G_m).
\end{equation}
With $S_K^{(oci)*}:\{G^*_1,...,G^*_K\}$, we can calculate the corresponding $OCI^*_K$:
\begin{equation}\label{eq:OCIK*}
OCI^*_K = \sum^K_{m=1}\sum^H_{h=1}OVL(\theta|Y_h,Z^*_m)\cdot I(\theta|Y_h\in G^*_m).
\end{equation}
An optimal clustering $S_K^{(km)*}$ can be found through the following K-Means algorithm:
\begin{equation}\label{eq:kmoci-1}
S_K^{(km)*} = \argmin_{S_K} \sum^K_{m=1}\sum^H_{h=1}d(\theta|Y_h,Z_m) \cdot I(\theta|Y_h\in G_m) 
\end{equation}
where $d(X,Y)=1-OVL(X,Y)$ is a distance measure between two random variables $X$ and $Y$. The equivalence between $S_K^{(oci)*}$ and $S_K^{(km)*}$ can be proved through the theorem as follows:
\begin{theorem}[Equivalence of clustering]\label{EqT}
The equivalence, $S_K^{(oci)*} = S_K^{(km)*}$, holds for any fixed $K$. 
\end{theorem}
\begin{proof}
\begin{equation*}
\begin{split}
S_K^{(oci)*} & = \argmin_{S_K} \sum^K_{m=1}\sum^H_{h=1}d(\theta|Y_h,Z_m) \cdot I(\theta|Y_h\in G_m) \\
& = \argmin_{S_K} \sum^K_{m=1}\sum^H_{h=1}[1-OVL(\theta|Y_h,Z_m)] \cdot I(\theta|Y_h\in G_m) \\
& = \argmax_{S_K} \sum^K_{m=1}\sum^H_{h=1}OVL(\theta|Y_h,Z_m)\cdot I(\theta|Y_h\in G_m) = S_K^{(km)*}.
\end{split}
\end{equation*}
\end{proof}

\subsection{Overlapping Evidence Index and Trade-off between Evidence Congruence and Robustness}\label{sec:OEIQSP}

For any fixed $K$, the heterogeneity among $p(\theta|Y_1),\ldots,p(\theta|Y_H)$ can be effectively identified through $S_K^{(oci)*}$. As demonstrated by the example in Section \ref{sec:sec2}, the evidence congruence increases as $K$ increases, and the information distortion reduces. However, the robustness of the synthesized prior weakens as $K$ increases because the amount of data in each cluster tends to decrease. It is desirable to find an optimal $K$ to strike a good balance for this trade-off. To help identifying the optimal $K$, we introduce the concept of an overlapping evidence index:

\begin{definition}[Overlapping Evidence Index]
Let $\theta|Y_1,\ldots,Y_H$ denote the random variable corresponding to the synthesized prior $\pi(\theta|Y_1,\ldots,Y_H)$. The overlapping evidence index (OEI) of $\pi(\theta|Y_1,\ldots,Y_H)$ is defined as the weighted sum of overlapping coefficients between the random variable $\theta|Y_1,\ldots,Y_H$ and each $\theta|Y_h$ for $h=1,\ldots,H$.
\begin{equation}\label{eq:oei}
     OEI(\pi)=\sum^H_{h=1}\frac{N_h}{N} \cdot OVL(\theta|Y_1,\ldots,Y_H,\theta|Y_h),
\end{equation} 
where $N_h$ is the sample size of $Y_h$ and $N=\sum^H_{h=1} N_h$.
\end{definition}

$OEI(\pi)$ lies within the interval $[0,1]$. It measures the consistency of synthesized prior $\pi(\theta|Y_1,\ldots,Y_H)$ with the information of external data (evidence). The higher the $OEI(\pi)$, the more congruent evidence transfers to $\pi(\theta|Y_1,\ldots,Y_H)$, from external data, and the less information distortion occurs. The details of choosing the optimal $K$ to balance the evidence congruence and robustness with the help of OEI will be introduced in Section \ref{sec:sec4}.

\section{Bayesian Clustering Prior}\label{sec:sec4}

For any fixed $K$, the optimal clustering map $S_K^{(oci)*}$ is denoted as follows: $$S_K^{(oci)*}:\{\theta|Y_1,\ldots,\theta|Y_H\} \rightarrow \{G_1,\ldots,G_K\}=\{\{\theta|Y_{11},\ldots, \theta|Y_{1n_1}\},\ldots,\{\theta|Y_{K1},\ldots, \theta|Y_{Kn_K}\}\}$$  where $n_m$ is the number of random variables in cluster $G_m$, $m=1,\ldots,K$. Based on $S_K^{(oci)*}$, the Bayesian clustering prior can be defined as follows:
\begin{definition}[Bayesian clustering prior]
The Bayesian clustering prior with $K$ clusters can be constructed as a weighted sum of informative priors synthesized from clusters, $G_m=\{Y_{m1},\ldots, Y_{mn_m}\}$, $m=1,\ldots,K$.
\begin{equation}\label{eq:bcp}
    \pi_K(\theta|Y_1,\ldots,Y_H) = \sum^K_{m=1}\frac{N_m}{N} \cdot \pi(\theta|Y_{m1},\ldots, Y_{mn_m})
\end{equation}
where $N_m$ is the number of observations in cluster $G_m$, and $N=\sum^K_{m=1}N_m$ is the size of all external data. Moreover, the prior in equation \eqref{eq:bcp} can be made more robust by adding a weighted weakly informative prior. We refer to this as the robust Bayesian clustering prior:
\begin{equation}\label{eq:rbcp}
    \pi^R_K(\theta|Y_1,\ldots,Y_H) = (1-w) \cdot \sum^K_{m=1}\frac{N_m}{N} \cdot \pi(\theta|Y_{m1},\ldots, Y_{mn_m}) + w \cdot \pi_0(\theta)
\end{equation}
where $\pi_0(\theta)$ is a weakly informative prior, $w\in [0,1]$ is the weight of $\pi_0(\theta)$.
\end{definition}

In equations \eqref{eq:bcp} and \eqref{eq:rbcp}, $\pi(\theta|Y_{m1},\ldots, Y_{mn_m})$ can be estimated through various methods, such as traditional Bayesian hierarchical models (BHM), MAP, or power priors. In this study, we choose MAP and robust MAP (rMAP) priors because they can be used in both trial design and data analysis stages. We refer to the resulting priors as Bayesian clustering MAP (BCMAP) and robust Bayesian clustering MAP (rBCMAP), respectively.
\cite{RBEST2021} developed the R package ``RBesT'' to implement the sampling of MAP and rMAP through Markov chain Monte Carlo (MCMC). To represent the MAP prior in parametric form, an expectation maximization (EM) algorithm is conducted to approximate the MCMC samples with a parametric mixture distribution. Since BCMAP (rBCMAP) is a weighted sum of MAPs, it can naturally be represented by a parametric mixture distribution as well. Thus, when conjugate MAP priors exist, a mixture of conjugate BCMAP priors can be used (see the real data example in Section \ref{sec:real_data}).

\begin{theorem}[Monotonic $OEIs$]\label{monoeis}
A BCMAP prior $\pi_K(\theta|Y_1,\ldots,Y_H)$ can be constructed based on $S_K^{(oci)*}$ for each $K$. Then, we can obtain a sequence of $OEI(\pi_K)$, $K=1,\ldots,H$, which monotonically increase as $K$ increases, i.e., $0 \leq OEI(\pi_1) \leq \ldots \leq OEI(\pi_H) \leq 1$. 
\end{theorem}
\begin{proof}
Without losing generality, let us compare $OEI(\pi_K)$ and $OEI(\pi_{K+1})$, $K<H$. Correspondingly, we need to consider the difference between $S_K^{(oci)*}$ and $S_{K+1}^{(oci)*}$. There must be a cluster in $S_K^{(oci)*}$, say cluster $G^{(K)}_m$, to be partitioned to two clusters in $S_{K+1}^{(oci)*}$, say $G^{(K+1)}_m$ and $G^{(K+1)}_{K+1}$, where $G^{(K)}_m=G^{(K+1)}_m\cup G^{(K+1)}_{K+1}$. In each cluster, we assume that the random variables $\theta|Y_h$ are homogeneous and use them to construct a MAP prior. For $G^{(K)}_m$, the corresponding MAP random variable is denoted as $\theta|G^{(K)}_m$. For $G^{(K+1)}_m$ and $G^{(K+1)}_{K+1}$, they are presented as $\theta|G^{(K+1)}_m$ and $\theta|G^{(K+1)}_{K+1}$. Since the homogeneity of cluster $G^{(K+1)}_m$ and cluster $G^{(K+1)}_{K+1}$ are greater than that of cluster $G^{(K)}_m$, it follows that:
\begin{equation*}
    \begin{split}
        \sum^{n_m}_{h=1} \frac{N_{mh}}{N} \cdot OVL(\theta|G^{(K)}_m,\theta|Y_{mh})\leq & \sum^{n_m}_{h=1} \frac{N_{mh}}{N} \cdot OVL(\theta|G^{(K+1)}_m,\theta|Y_{mh})\cdot I(Y_{mh} \in G^{(K+1)}_m) \\
         + & \sum^{n_m}_{h=1} \frac{N_{mh}}{N} \cdot OVL(\theta|G^{(K+1)}_{K+1},\theta|Y_{mh})\cdot I(Y_{mh} \in G^{(K+1)}_{K+1}).
    \end{split}
\end{equation*}
By equation \eqref{eq:oei} in the definition of OEI, the inequality $OEI(\pi_K)\leq OEI(\pi_{K+1})$ holds. Note: This proof is applied to the rBCMAP as well.
\end{proof}

We can scale the sequence $OEI(\pi_1) \leq \ldots \leq OEI(\pi_H)$ by the maximum $OEI(\pi_H)$ and refer to $OEI(\pi_K)/OEI(\pi_H)$, $K=1,\ldots,H$, as $SOEI(\pi_K)$. The sequence $SOEI(\pi_K)=\{OEI(\pi_1)/OEI(\pi_H),\ldots,OEI(\pi_{H-1})/OEI(\pi_H),1\}$ denotes the percentage of evidence congruence under each $K$ relative to the extreme case where each distribution is a cluster. A threshold balancing the trade-off of maximizing evidence congruence and minimizing the number of clusters $K$ for robustness can then be used to determine the optimal $K$. As a rule of thumb, we recommend to choose the threshold at 60\% for selecting the optimal $K$, which indicates favoring congruence slightly more but not to the extent that robustness is significantly compromised. With this threshold, we can find the minimal $K^*$ that satisfies $OEI(\pi_{K^*})/OEI(\pi_H)\geq 60\%$. Then, $\pi_{K^*}(\theta|Y_1,\ldots,Y_H)$ is the finalized optimal Bayesian clustering prior. The main advantage of this threshold-based clustering approach is twofold: (1) it is straightforward and easy to interpret, as the threshold reflects a balance between the congruence of the synthesized prior with the external data and its robustness to future data, and (2) a fixed threshold offers a consistent and unified interpretation across various applications. These contrasts with Bayesian nonparametric clustering approaches, where the value of hyperparameters used to determine the number of clusters are often challenging to select and interpret and may lack consistency across different applications.

\begin{figure}[t!]
\centering
\includegraphics[width=.9\textwidth]{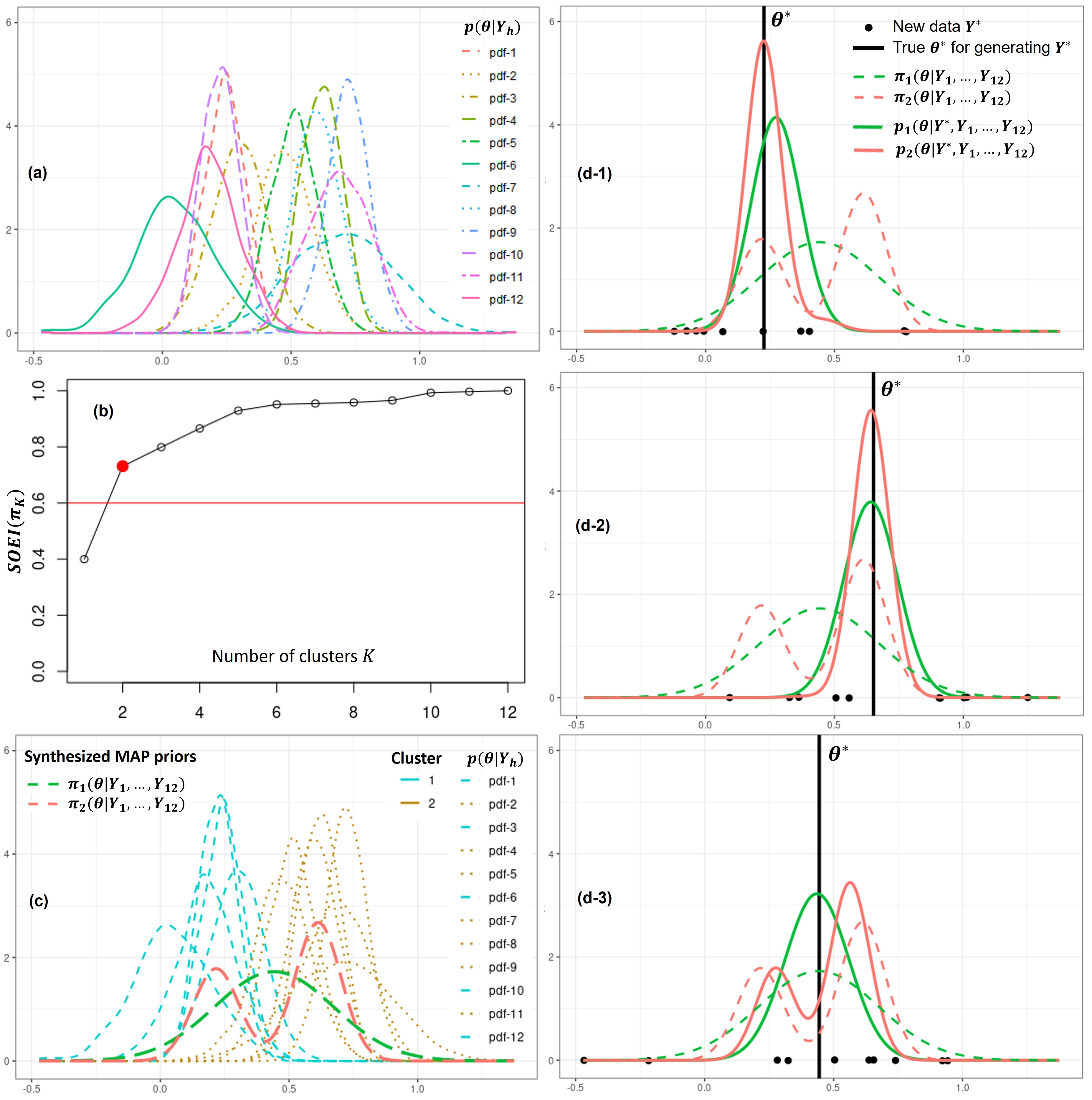}
\caption[bcp]{An example of BCMAP and comparison with MAP. Panel (a) shows the posteriors $p(\theta|Y_h)\propto L(Y_h|\theta)\pi_0(\theta)$ with weakly informative prior $\pi_0(\theta)$, $h=1,\ldots,12$. Panel (b) presents the identification of optimal number of clusters $K^*=2$. Panel (c) exhibits the clustering of external posteriors and corresponding BCMAP prior $\pi_2(\theta|Y_1,\ldots, Y_{12})$. MAP prior $\pi_1(\theta|Y_1,\ldots, Y_{12})$ is provided for comparison. Panels (d-1), (d-2), and (d-3) show the posteriors, $p_i(\theta|Y^*,Y_1,\ldots, Y_{12})\propto L(Y^*|\theta)\pi_i(\theta|Y_1,\ldots, Y_{12})$, $i=1,2$, with three sets of new data $Y^*$ generated from $N(\theta^*,0.3^2)$, where $\theta^*$ is sampled from model \eqref{eq:example}.}
\label{fig:mapexp}
\end{figure}

The following example provides an intuitive understanding of the questions: how is BCMAP constructed and how and why does it work? In the example, 12 external datasets were generated by the following model \eqref{eq:example}. The summary of these datasets is provided in Table \ref{tb:externalData} in Appendix. It is easy to find the heterogeneity among these datasets that contains two clusters centered at 0.2 and 0.6.
\begin{equation}\label{eq:example}
\begin{split}
& Y_1 | \theta,\sigma^2,\ldots, Y_{12} | \theta,\sigma^2 \sim N(\theta,\sigma^2), \hspace{3mm} \sigma \sim Halfnormal(1),\\
& \theta =\phi*(\mu_{c1}+\epsilon_{c1})+(1-\phi)*(\mu_{c2}+\epsilon_{c2}), \\
& \phi \sim Bernoulli(p=0.5), \hspace{3mm}  \mu_{c1}=0.2, \hspace{3mm} \mu_{c2}=0.6, \\
&  \epsilon_{c1} \sim N(0,\sigma_{c1}^2), \hspace{3mm} \epsilon_{c2} \sim N(0,\sigma_{c2}^2), \hspace{3mm} \sigma_{c1}=0.05, \hspace{3mm} \sigma_{c2}=0.08.
\end{split}
\end{equation}

The posteriors $p(\theta|Y_h)\propto L(Y_h|\theta)\pi_0(\theta)$, $h=1,\ldots,12$, are shown in Panel (a) of Figure \ref{fig:mapexp}. The weakly informative prior $\pi_0(\theta)$ is constructed as: $\pi_0(\theta) = \int \pi(\theta|\tau)\pi(\tau)\,d\tau$, where $\pi(\theta|\tau)\sim N(0.4,1/\tau)$ and $\pi(\tau)\sim gamma(0.01,10)$. In Panel (b), according to $SOEI(\pi_K)$ and the threshold 60\%, we identify the optimal number of clusters as $K^*=2$. The corresponding BCMAP prior, $\pi_2(\theta|Y_1,\ldots,Y_{12})$, is presented in Panel (c). For comparison, we also provide the MAP prior constructed from one cluster, $\pi_1(\theta|Y_1,\ldots,Y_{12})$. The BCMAP prior is congruent with external data, i.e., RWD. Therefore, if the new data $Y^*$ is consistent with the RWD too, which is reasonable, the estimation of $p(\theta|Y^*)$ will benefit from the information borrowing with the BCMAP prior. To illustrate this, three examples are shown in Panels (d-1), (d-2), and (d-3). Correspondingly, three different new datasets $Y^*$ (black dots) are generated from $N(\theta^*,0.3^2)$, where $\theta^*$ (black vertical line) is generated following the $\theta$ defined in \eqref{eq:example}, maintaining consistency with the external data. We compare the performance of BCMAP and MAP in estimating $\theta^*$.
In Panel (d-1), $\theta^*$ is near the center at 0.2. The posterior with BCMAP prior outperforms MAP prior in the estimations of both location (bias) and scale (variance). In Panel (d-2), $\theta^*$ is near the center at 0.6. The posteriors with both BCMAP and MAP priors offer accurate location estimation, but the BCMAP provides a better scale estimation. Note that BCMAP does not always perform better than MAP. In Panel (d-3), the posterior with the BCMAP prior shows two modes and a wider scale compared to the posterior with the MAP prior. The reason is that $\theta^*$ appears in the valley of the BCMAP prior but near the peak of the MAP prior, which is the opposite of the reason BCMAP outperforms MAP in Panels (d-1) and (d-2). The key point is that the likelihood of $\theta^*$ appearing near the two centers, 0.2 and 0.6, is much higher than it appearance near the valley according to RWD.

The example highlights the strengths and limitations of the proposed Bayesian clustering prior. In summary, the superiority of this approach depends on two key conditions, which are commonly met in practice.
\begin{itemize}
    \item[(1)] \textbf{Heterogeneity among external datasets:} The Bayesian clustering MAP (BCMAP) prior is particularly advantageous when there is heterogeneity among the external datasets. In the absence of such heterogeneity, the BCMAP reduces to the standard MAP prior (see Section 5.1).
    \item[(2)] \textbf{Consistency between future data and external data:} For the BCMAP to perform well, the future (new) data $Y^*$ must be consistent with the external data. Specifically, the data generating process for $Y^*$ should align with that of the external datasets, even if it is not identical. For example, $Y^*$ may follow a distribution that corresponds to one of the components in the mixture distribution of the external data. If this condition is not met, the BCMAP may underperform compared to the MAP prior; however, the rBCMAP can provide a partial remedy by incorporating robustness into the prior construction (see Section 5.1).
\end{itemize}
\section{Simulation Studies}\label{sec:sim}
In this section, we conduct comprehensive simulation studies to demonstrate the superiority of BCMAP (rBCMAP) in both parameter estimation and hypothesis testing compared to commonly used priors, such as MAP, rMAP, NPP (Normalized Power Prior), and MEM. For both rBCMAP and rMAP, the weight for robustness (weight of the component of weakly informative prior) is 0.5.

\subsection{Parameter Estimation}

\begin{figure}[t!]
\centering
\includegraphics[width=.72\textwidth]{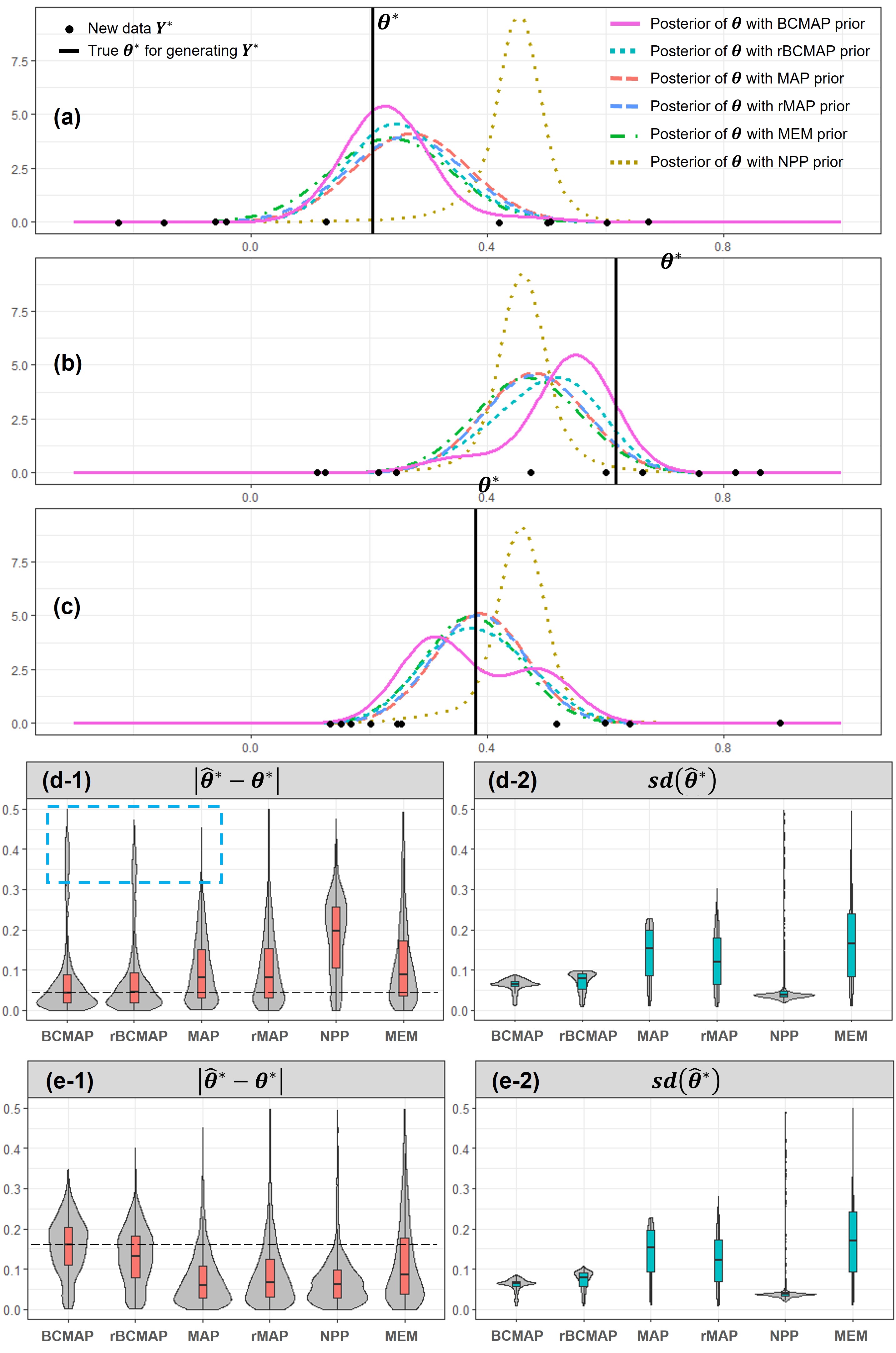}
\caption[bcp]{Estimation of new $\theta^*$ from posteriors computed with different priors. In Panels (a) and (b), $\theta^*$ is located near the centers 0.2 and 0.6, respectively, while in Panel (c), $\theta^*$ lies between the two centers. Panels (d-1, d-2) present the bias and standard deviation of the estimate $\hat{\theta}^*$ when $\theta^*$ is near the centers 0.2 or 0.6. Panels (e-1, e-2) display the bias and standard deviation when $\theta^*$ is situated between the two centers. Both (d-1, d-2) and (e-1, e-2) are based on 1000 simulations, with each simulation containing 10 observations of $Y^*$.}
\label{fig:simExamples}
\end{figure}

First, let us gain an intuitive understanding of the simulation using the 12 external datasets shown in Panel (a) of Figure \ref{fig:mapexp}. In each simulation, 10 observations are generated from $N(\theta^*, 0.3^2)$, where $\theta^*$  is drawn from the model \eqref{eq:example}. The posteriors are then calculated based on $Y^*$ using different priors.
Intuitive comparison can be conducted by considering the posterior estimation of the location (bias) and scale (standard deviation) of estimate $\hat{\theta}^*$. 
Let us examine three examples shown in panels (a, b, c) of Figure \ref{fig:simExamples}. In panels (a) and (b), $\theta^*$ is close to the centers 0.2 or 0.6, indicating the consistency of the new data $Y^*$ with the external data. BCMAP and rBCMAP outperform other methods with less bias and lower variance of the estimate $\hat{\theta}^*$. However, panel (c) exhibits the opposite scenario where $\theta^*$ is located in the middle of the two centers, indicating the inconsistency between the new data $Y^*$ and the external data. This results in a significantly worse performance of BCMAP compared to the other methods. Such scenarios as in panel (c) are rare under the assumption that the new data $Y^*$ is consistent with the external data, i.e., RWD. A valuable observation in panel (c) is that rBCMAP can enhance the robustness of BCMAP when the assumption of consistency is violated. Furthermore, we repeated the simulation 1000 times for scenarios that either adhere to or violate the assumption of consistency. The result of the scenarios adhering to the assumption is shown in panels (d-1,2). Panel (d-1) presents the empirical distribution of bias, $|\hat{\theta^*}-\theta^*|$, where $\hat{\theta^*}$ is the mode of the posterior. BCMAP and rBCMAP have less bias compared to other methods. Panel (d-2) shows the empirical distribution of standard deviation, $sd(\hat{\theta}^*)$. Although NPP has the lowest standard error, it has the highest bias in Panel (d-1). Excluding NPP, BCMAP and rBCMAP perform better than other methods. In inconsistency scenarios, BCMAP exhibits the highest bias compared to other methods, as shown in Panel (e-1). With a robustness weight of 0.5, rBCMAP reduces the estimation bias, which justifies the observation in Panel (c). The summary of the estimation results based on the root mean square error (RMSE) criterion for both the consistency and inconsistency scenarios is presented in Table \ref{tb:fig3}. An unexpected observation is that the BCMAP exhibits a higher RMSE (0.141) compared to MAP (0.129) in the consistency scenario, which appears to contradict the results shown in Figure \ref{fig:simExamples}. This discrepancy is due to the influence of outliers, as highlighted in Panel (d-1) of Figure \ref{fig:simExamples} (indicated by the blue dashed rectangle). After removing the top 5\% of data (under bias measure), the BCMAP's RMSE (0.109) becomes smaller than that of MAP (0.113), aligning with the results shown in Panel (d-1). Another noteworthy observation is that the rBCMAP effectively mitigates the RMSE inflation observed with BCMAP in both scenarios.

\begin{table}[t!]
\caption{Comparison of RMSE under the scenarios of consistency and inconsistency shown in Panel (d-1), (d-2) and (e-1), (e-2) in Figure \ref{fig:simExamples}. ``Remove 5\% outliers'' means that calculates RMSE after removing the top 5\% of data (under bias measure).}
\centering
\small
\begin{tabular}{|c|c|c|c|c|c|c|c|}
\hline
\multirow{2}{*}{Scenarios} & \multirow{2}{*}{ Data} & \multicolumn{6}{c|}{Root mean square error (RMSE)}                                                                                                                      \\ \cline{3-8} 
                             &                       & \multicolumn{1}{c|}{BCMAP}  & \multicolumn{1}{c|}{MAP}    & \multicolumn{1}{c|}{rBCMAP} & \multicolumn{1}{c|}{rMAP}   & \multicolumn{1}{c|}{NPP}    & MEM    \\ \hline
\multirow{2}{*}{Consistency}   & All                       & 0.141 & 0.129 & 0.131 & 0.150 & 0.285 & 0.209 \\ \cline{2-8} 
                              & Remove 5\% outliers & 0.109 & 0.113 & 0.103 & 0.119 & 0.207 & 0.149 \\ \hline \hline
\multirow{2}{*}{Inconsistency} & All                       & 0.170 & 0.107 & 0.149 & 0.133 & 0.242 & 0.238 \\ \cline{2-8} 
                              & Remove 5\% outliers & 0.162 & 0.085 & 0.139 & 0.099 & 0.098 & 0.157 \\ \hline 
\end{tabular}
\label{tb:fig3}
\end{table}

The simulation illustrated in Figure \ref{fig:simExamples} considers only a single heterogeneity scenario (two clusters) in the external datasets. To enable a more comprehensive comparison, we extend the study in two aspects: (1) in addition to the two-cluster scenario, we examine cases with one and three clusters; (2) we evaluate the effect of varying the size of $Y^*$, considering sizes of 5, 10, 15, 20, 25, and 30. In the one-cluster scenario, 10 external datasets are generated from the model \eqref{eq:example1}, representing homogeneity among the datasets. While the three-clusters scenario includes 25 external datasets generated from the model \eqref{eq:example3}, reflecting different heterogeneity from the two-clusters scenario. A summary of the three scenarios, including the number of clusters, number of datasets, and sample sizes within each dataset, is provided in Table \ref{tb:simParaEst} in the Appendix. The posteriors (based on a weakly informative prior), clustering results, and synthesized BCMAP and rBCMAP priors for the one-cluster and three-cluster scenarios are presented in Figures \ref{fig:simOneCluster} and \ref{fig:simThreeClusters}, respectively, in the Appendix.

\begin{equation}\label{eq:example1}
\begin{split}
& Y_1 | \theta,\sigma^2,\ldots, Y_{10} | \theta,\sigma^2 \sim N(\theta,\sigma^2), \hspace{1mm} \sigma \sim Halfnormal(1), \\
& \theta =\mu+\epsilon, \hspace{1mm} \mu=0.6, \hspace{1mm} \epsilon \sim N(0,0.12^2).
\end{split}
\end{equation}

\begin{equation}\label{eq:example3}
\begin{split}
& Y_1 | \theta,\sigma^2,\ldots, Y_{25} | \theta,\sigma^2 \sim N(\theta,\sigma^2), \hspace{3mm} \sigma \sim Halfnormal(1), \hspace{3mm} \theta =\phi*(\mu_{c1}+\epsilon_{c1}), \\
& \theta =I(c=1)*(\mu_{c1}+\epsilon_{c1})+I(c=2)*(\mu_{c2}+\epsilon_{c2})+I(c=3)*(\mu_{c3}+\epsilon_{c3}), \\
& c \sim multinomial(p_{c1}=p_{c2}=0.3,p_{c3}=0.4), \hspace{3mm}  \mu_{c1}=0.2, \hspace{3mm} \mu_{c2}=0.6, \hspace{3mm} \mu_{c3}=1, \\
&  \epsilon_{c1} \sim N(0,0.12^2), \hspace{3mm} \epsilon_{c2} \sim N(0,0.1^2), \hspace{3mm} \epsilon_{c3} \sim N(0,0.15^2).
\end{split}
\end{equation}

\begin{figure}[t!]
\centering
\includegraphics[width=\textwidth]{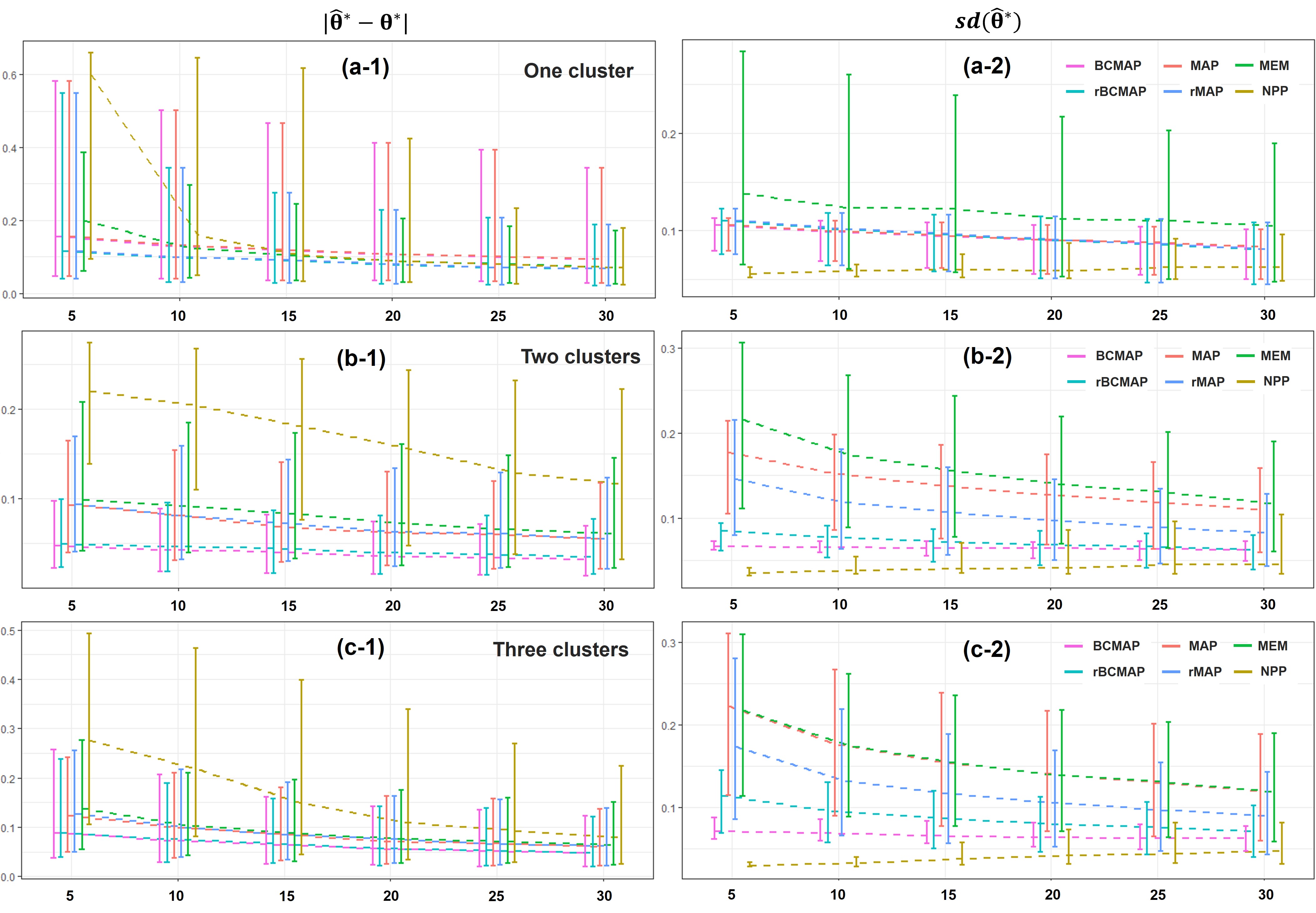}
\caption[bcp]{Comparison of parameter estimation under different size of observations. Panels (a) and (b) present the bias and standard deviation of estimate $\hat{\theta}^*$ with different numbers of observations $Y^*$, when $\theta^*$ is near the centers 0.2 or 0.6.}
\label{fig:sim}
\end{figure}

The simulation results are presented in Figure \ref{fig:sim}. The panels in the left and right columns illustrate the comparison of bias $|\hat{\theta^*}-\theta^*|$ and variance $sd(\hat{\theta^*})$, respectively. In the one-cluster scenario, the MAP and rMAP priors are identical to the BCMAP and rBCMAP priors, respectively, as shown in Panels (a-1) and (a-2). In the two- and three-clusters scenarios, when the new data $Y^*$ is consistent with the external data, the BCMAP and rBCMAP priors outperform other methods.
Notably, the NPP exhibits poor performance, characterized by the smallest $sd(\hat{\theta^*})$ but the largest bias. An interesting observation is that as the size of $Y^*$ increases, all methods tend to converge to similar results because the new data begins to dominate the prior. The summary of RMSE for the two clusters and three clusters scenarios (excluding the one cluster scenario, as BCMAP and rBCMAP are equivalent to MAP and rMAP in this case) is presented in Table \ref{tb:2clusters} and Table \ref{tb:3clusters}, respectively. As with Table \ref{tb:fig3}, we report the models' RMSE values both with and without outliers.

In summary, if the new data $Y^*$ is consistent with external data that exhibits heterogeneity, BCMAP and rBCMAP priors provide more accurate parameter estimation, with lower bias and variance, than commonly used methods. When the new data $Y^*$ is not consistent with the external data, both BCMAP and rBCMAP may underperform compared to other methods. However, rBCMAP provides a partial remedy, as demonstrated in Panel (e-1) of Figure \ref{fig:simExamples} and Tables \ref{tb:fig3}, \ref{tb:2clusters}, \ref{tb:3clusters}.

\subsection{Hypothesis Testing}
In this section, we utilize the external datasets with two clusters shown in Figure \ref{fig:mapexp}.
To compare the performance of different priors in hypothesis testing, we design a trial with two groups: control and treatment. Observations are generated from $N(\theta_c, \sigma^2)$ and $N(\theta_t, \sigma^2)$, where $\sigma=0.3$ and $\theta_t=0.4$. Since the possible values of $\theta_c$ are centered at 0.2 and 0.6 by the model \eqref{eq:example}, we consider two scenarios: $\theta_{c1}=0.2$ and $\theta_{c2}=0.6$. Correspondingly, we are interested in two hypothesis tests:
\begin{itemize}
\centering
    \item[(1)] $H_0: \theta_t \leq \theta_{c1} , \hspace{3mm} H_1: \theta_t > \theta_{c1}$
    \item[(2)] $H_0: \theta_t \geq \theta_{c2} , \hspace{3mm} H_1: \theta_t < \theta_{c2}$.
\end{itemize}

In the simulation, we investigate both frequentist and Bayesian methods. The frequentist method is a one-sided two-sample t-test. We considered two control vs. treatment recruitment ratios, $10:30$ and $30:30$. For Bayesian methods, to study the effect of information borrowing from synthesized priors, we use the same data with control vs. treatment ratio $10:30$. Corresponding to the two hypothesis tests, the decision rule (reject $H_0$) and operational characteristics in Bayesian methods are as follows:
\begin{itemize}
    \item[(1)] Decision rule:    $Pr(\theta_t>\theta_{c1}|Y_t^*,Y_{c1}^*,Y_1,\ldots,Y_{12})>\eta$; \\ Type1 error rate: $Pr(Pr(\theta_t>\theta_{c1}|Y_t^*,Y_{c1}^*,Y_1,\ldots,Y_{12})>\eta| H_0)$;\\
    Power:        $Pr(Pr(\theta_t>\theta_{c1}|Y_t^*,Y_{c1}^*,Y_1,\ldots,Y_{12})>\eta| H_1)$.
    \item[(2)] Decision rule: $Pr(\theta_t<\theta_{c2}|Y_t^*,Y_{c2}^*,Y_1,\ldots,Y_{12})>\eta$; \\ Type1 error rate: $Pr(Pr(\theta_t<\theta_{c2}|Y_t^*,Y_{c2}^*,Y_1,\ldots,Y_{12})>\eta| H_0)$;\\
    Power:        $Pr(Pr(\theta_t<\theta_{c2}|Y_t^*,Y_{c2}^*,Y_1,\ldots,Y_{12})>\eta| H_1)$,
\end{itemize}
where $\eta \in [0.5,1)$ is an adjustable threshold used to control the type 1 error rate under 5\%. In the calculation of $Pr(\theta_t>\theta_{c1}|Y_t^*,Y_{c1}^*,Y_1,\ldots,Y_{12})$ or $Pr(\theta_t<\theta_{c2}|Y_t^*,Y_{c2}^*,Y_1,\ldots,Y_{12})$, we need to find the joint posterior $p(\theta_t, \theta_c |Y_t^*,Y_c^*,Y_1,\ldots,Y_{12})$, where $\theta_c=\{\theta_{c1},\theta_{c2}\}$ and $Y_c^*=\{Y_{c1}^*,Y_{c2}^*\}$. Since $\theta_t$ and $\theta_c$ are independent, the joint posterior can be expressed as follows:
\begin{equation}\label{eq:hpt}
        p(\theta_t, \theta_c |Y_t^*,Y_c^*,Y_1,\ldots,Y_{12}) \propto L(Y_t^*|\theta_t)\pi_0(\theta_t) L(Y_c^*|\theta_c)\pi(\theta_c|Y_1,\ldots,Y_{12})
\end{equation}
where $\pi_0(\theta_t)$ is a weakly informative prior indicating trivial prior knowledge about the treatment group, and $\pi(\theta_c|Y_1,\ldots,Y_{12})$ denotes the synthesized priors from the 12 external datasets.

\begin{table}[t!]
\centering
\small
\caption{Comparison of operation characteristic in hypothesis testing. Two hypothesis tests: (1) $H_0: \theta_t \leq \theta_{c1} , \hspace{3mm} H_1: \theta_t > \theta_{c1}$ and (2) $H_0: \theta_t \geq \theta_{c2} , \hspace{3mm} H_1: \theta_t < \theta_{c2}$. For the Bayesian methods, the control vs. treatment ratio is $10:30$.}
\begin{tabular}{|ccccc|}
\hline
\multicolumn{1}{|c|}{\multirow{2}{*}{\textbf{Methods}}} & \multicolumn{2}{c|}{\textbf{$\theta_{c1}=0.2$}}             & \multicolumn{2}{c|}{\textbf{$\theta_{c2}=0.6$}} \\ \cline{2-5} 
\multicolumn{1}{|c|}{}                                  & \textbf{Type 1 Error} & \multicolumn{1}{c|}{\textbf{Power}} & \textbf{Type 1 Error}      & \textbf{Power}     \\ \hline \hline
\multicolumn{5}{|c|}{\textbf{Control : Treatment = 30 : 30}}                                                                                                            \\ \hline
\multicolumn{1}{|c|}{\textbf{Frequentist}}              & 0.044                 & \multicolumn{1}{c|}{0.810}          & 0.055                      & 0.817              \\ \hline \hline
\multicolumn{5}{|c|}{\textbf{Control : Treatment = 10 : 30}}                                                                                                            \\ \hline
\multicolumn{1}{|c|}{\textbf{Frequentist}}              & 0.055                 & \multicolumn{1}{c|}{0.559}          & 0.053                      & 0.550              \\ \hline
\multicolumn{1}{|c|}{\textbf{MEM}}                      & 0.051                 & \multicolumn{1}{c|}{0.601}          & 0.052                      & 0.581              \\
\multicolumn{1}{|c|}{\textbf{NPP}}                      & 0.001                 & \multicolumn{1}{c|}{0.275}          & 0.010                      & 0.886              \\
\multicolumn{1}{|c|}{\textbf{MAP}}                      & 0.052                 & \multicolumn{1}{c|}{0.631}          & 0.054                      & 0.629              \\
\multicolumn{1}{|c|}{\textbf{rMAP}}                     & 0.051                 & \multicolumn{1}{c|}{0.607}          & 0.055                      & 0.605              \\
\multicolumn{1}{|c|}{\textbf{BCMAP}}                    & \textbf{0.049}        & \multicolumn{1}{c|}{\textbf{0.797}} & \textbf{0.052}             & \textbf{0.781}     \\
\multicolumn{1}{|c|}{\textbf{rBCMAP}}                   & 0.050                 & \multicolumn{1}{c|}{0.693}          & 0.051                      & 0.709              \\ \hline
\end{tabular}
\label{tb:sim}
\end{table}

The simulation results are shown in Table \ref{tb:sim}. For the frequentist method, the results of both hypothesis tests (1) and (2) show a dramatic reduction in power as the size of the control group decreases from 30 to 10. Regarding Bayesian methods, except for the case of NPP in hypothesis test (1), information borrowing from synthesized priors can improve the test power compared to the frequentist 10:30 trial, which does not involve information borrowing. However, the simulation results also illustrate that the quality of the prior has a significant impact. NPP performs best in test (2) but worst in test (1). MEM and rMAP only improve the power in small magnitudes (less than 5\%). MAP contributes a considerable power increase in both tests (greater than 7\%), but still smaller than rBCMAP and BCMAP. Overall, BCMAP performs the best (increasing power more than 20\%), even comparable to the frequentist 30:30 trial. In sum, Bayesian cluster priors can provide accurate information about external data, especially when the external data exhibits heterogeneity. When the new data $Y^*_c$ is generally consistent with the external data, borrowing information from Bayesian cluster priors can effectively improve the operational characteristics of hypothesis testing.

\section{Real Data Analysis}
\label{sec:real_data}
Postoperative nausea and vomiting are common complications following surgery and anesthesia. As an alternative to drug therapy, acupuncture has been studied as a potential treatment in several trials \citep{Lee04}. The dataset ``dat.lee2004" in the R package ``metadat" \citep{metadat2022} contains the results from 16 clinical trials examining the effectiveness of wrist acupuncture point P6 treatment for preventing postoperative nausea. Patient level (covariate) information is not available. The detailed description of the dataset can be found in Table \ref{tb:real} in the Appendix.

Based on this dataset, we aim to construct an informative prior that contains accurate real-world information about the wrist acupuncture point P6 treatment. To do this, we focus on the treatment groups listed in Table \ref{tb:real}. The columns ``ai'' and ``n1i'' correspond to $r_h$ and $N_h$, respectively, where $r_h \sim Binomial(N_h, p)$ for $h=1,\ldots,16$, with $p$ representing the probability of a patient experiencing nausea. Let $Y_h=\{r_h, N_h\}$ denote the data of study $h$. Our goal is to construct an informative prior $\pi(p|Y_1,\ldots,Y_{16})$. First, we examine the posteriors of the external data using a vague conjugate prior $\pi(p)=Beta(0.5,0.5)$. The posteriors can be straightforwardly obtained: $p | Y_h \sim Beta(0.5+y_{h}, 0.5+N_h-y_h), \hspace{3mm} h=1,\ldots,16$, as shown in Panel (b) of Figure \ref{fig:realData}. Clearly, there is substantial heterogeneity among the different studies, making it a suitable case for leveraging the Bayesian clustering prior.
From Panel (a), we identified the optimal number of clusters as $K^*=3$. The clusters and the corresponding BCMAP prior are presented in Panel (b), where we also include the MAP prior for visual comparison. The parameterized MAP and BCMAP are listed as follows:
\begin{equation}\label{eq:realMAP}
\pi_{MAP}(p|Y_1,\ldots,Y_{16}) =\pi_1(p|Y_1,\ldots,Y_{16}) =Beta(1.7,4.0) \hspace{21mm}
\end{equation}
\begin{equation}\label{eq:realBCMAP}
    \begin{split}
        & \pi_{BCMAP}(p|Y_1,\ldots,Y_{16}) = \pi_3(p|Y_1,\ldots,Y_{16}) \\
        &=0.18*Beta(3.7,43.2)+0.47*Beta(11.2,43.2)+0.35*Beta(7.3,8.1).
    \end{split}
\end{equation}

\begin{figure}[t!]
\centering
\includegraphics[width=\textwidth]{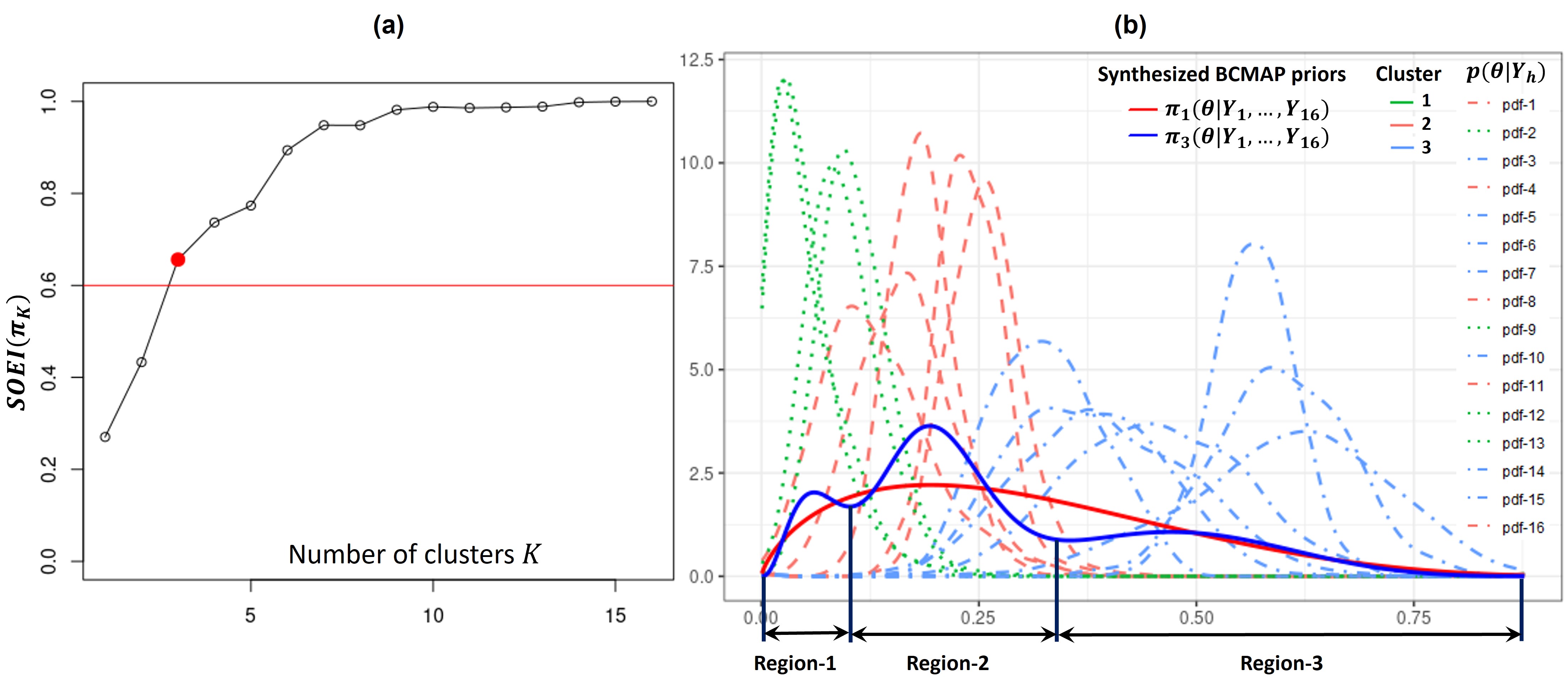}
\caption[bcp]{Bayesian clustering prior constructed on the wrist acupuncture point P6 treatment data from 16 studies. Panel (a) shows the selection of optimal number of clusters $K^*=3$. Panel (b) presents the clustering results and corresponding parameterized BCMAP prior $\pi_3(p|Y_1,\ldots,Y_{16})$. As a comparison, the parameterized MAP prior $\pi_1(p|Y_1,\ldots,Y_{16})$ is also shown in Panel (b).}
\label{fig:realData}
\end{figure}

Since the MEM and NPP need to compare the similarity between external data and new data, they cannot be explicitly constructed without knowing the new data. In contrast, MAP (rMAP) and BCMAP (rBCMAP) are independent of new data. Compared to the MAP prior, the BCMAP prior has the advantage that its effective sample size (ESS) can be estimated in a finer way. From equation \eqref{eq:realMAP}, we know that the ESS of $\pi_{MAP}(p|Y_1,\ldots,Y_{16})$ over the entire region [0,1] is 1.7+4.0=5.7. Let us check the ESS of $\pi_{BCMAP}(p|Y_1,\ldots,Y_{16})$ from equation \eqref{eq:realBCMAP} and Panel (b): in region 1, $ESS_1 \approx 0.18*(3.7+43.2)=8.4$; in region 2, $ESS_2 \approx 0.47*(11.2+43.2)=25.6$; in region 3, $ESS_3 \approx0.35*(7.3+8.1)=5.4$. This provides valuable information for the design of new trials. 
In the data analysis stage, the amount of information borrowed from $\pi_{BCMAP}(p|Y_1,\ldots,Y_{16})$ varies across different regions determined by the clustering result. For example, the borrowed information will be significantly stronger when $\hat{p}=r^*/N^*$ falls into region-2 compared to when $\hat{p}$ falls into region-1 or region-3. This enhances the congruence of the prior with the evidence, especially in the presence of heterogeneous external data. Figure \ref{fig:realData2} in Appendix illustrates the scenario with six clusters, where $\pi_6$ is another option of the informative prior.

\section{Discussion}\label{sec:Con}
Clustering is essential for synthesizing an informative prior from heterogeneous multisource external data. Leveraging the concept of the overlapping coefficient, we introduced the OCI, OEI, and a K-Means algorithm to identify the optimal clustering based on a criterion that balances the trade-off between evidence congruence and the robustness of the synthesized prior. Using the clustering results, a Bayesian clustering prior can be constructed and applied during both the trial design and data analysis stages. Simulation studies validate its superiority over commonly used priors when heterogeneity exists among data sources and the future (new) data is consistent with the external data.

In Section \ref{sec:sec4}, we recommended using a 60\% threshold, which is a subjective suggestion. Researchers can adjust this threshold based on their perspective on the trade-off between evidence congruence and robustness. Those who prioritize robustness might opt for a 50\% threshold, while those who prioritize evidence congruence might choose a threshold of 70\%. Basically, the choice of the threshold reflects the prior knowledge about the future and external data and the preference of the congruence and robustness trade-off.

We do not include covariate information in our study. However, it is becoming increasingly important in the era of precision medicine. The possible ways to incorporate covariate information into the Bayesian cluster prior is an interesting direction for future study.


\newpage

\begin{appendix}
\renewcommand\thefigure{\thesection.\arabic{figure}}
\renewcommand\thetable{\thesection.\arabic{table}}

\section{Appendix}\label{sec:apdxA}

\begin{table}[ht]
\caption{Summary of 12 external datasets (two clusters).}
\small
\centering
\begin{tabular}{ccccc}
\hline
Dataset ($Y_h$) & Size ($n_h$) & Cluster center ($\mu_c$) & Mean ($\theta_h$)  & Standard deviation ($\sigma_h$) \\ \hline
    $Y_1$    & 88   & 0.2            & 0.220 & 0.560              \\
    $Y_2$    & 50   & 0.6            & 0.465 & 0.230              \\
    $Y_3$    & 56   & 0.2            & 0.172 & 1.559              \\
    $Y_4$    & 99   & 0.6            & 0.612 & 0.071              \\
    $Y_5$    & 69   & 0.6            & 0.509 & 0.129              \\
    $Y_6$    & 25   & 0.2            & 0.102 & 1.715              \\
    $Y_7$    & 15   & 0.6            & 0.634 & 0.461              \\
    $Y_8$    & 81   & 0.6            & 0.576 & 1.265              \\
    $Y_9$    & 95   & 0.6            & 0.672 & 0.687              \\
    $Y_{10}$    & 95   & 0.2            & 0.189 & 0.446              \\
    $Y_{11}$    & 48   & 0.6            & 0.666 & 1.224              \\
    $Y_{12}$    & 40   & 0.2            & 0.164 & 0.360              \\ \hline
\end{tabular}
\label{tb:externalData}
\end{table}

\begin{table}[ht]
\caption{Comparison of RMSE under two clusters scenario with different number of $Y^*$. ``Remove 5\% outliers'' means that calculates RMSE after removing the top 5\% of data (under bias measure).}
\small
\begin{tabular}{|c|c|c|c|c|c|c|c|}
\hline
\multirow{2}{*}{Number of $Y^*$} & \multirow{2}{*}{ Data } & \multicolumn{6}{c|}{Root mean square error (RMSE)}                                                                                                                                     \\ \cline{3-8} 
                             &                       & \multicolumn{1}{c|}{BCMAP}  & \multicolumn{1}{c|}{MAP}    & \multicolumn{1}{c|}{rBCMAP} & \multicolumn{1}{c|}{rMAP}   & \multicolumn{1}{c|}{NPP}    & MEM    \\ \hline
\multirow{2}{*}{5}  & All & 0.150 & 0.137 & 0.135 & 0.168 & 0.361 & 0.241 \\ \cline{2-8} 
                    & Remove 5\% outliers  & 0.122 & 0.122 & 0.107 & 0.127 & 0.220 & 0.161 \\ \hline
\multirow{2}{*}{10} & All & 0.144 & 0.131 & 0.135 & 0.151 & 0.282 & 0.209 \\ \cline{2-8} 
                    & Remove 5\% outliers  & 0.112 & 0.114 & 0.107 & 0.119 & 0.207 & 0.149 \\ \hline
\multirow{2}{*}{15} & All & 0.138 & 0.123 & 0.130 & 0.137 & 0.238 & 0.185 \\ \cline{2-8} 
                    & Remove 5\% outliers  & 0.105 & 0.106 & 0.099 & 0.112 & 0.192 & 0.142 \\ \hline
\multirow{2}{*}{20} & All & 0.129 & 0.116 & 0.120 & 0.128 & 0.223 & 0.174 \\ \cline{2-8} 
                    & Remove 5\% outliers  & 0.094 & 0.098 & 0.087 & 0.103 & 0.178 & 0.130 \\ \hline
\multirow{2}{*}{25} & All & 0.123 & 0.112 & 0.114 & 0.123 & 0.208 & 0.164 \\ \cline{2-8} 
                    & Remove 5\% outliers  & 0.086 & 0.093 & 0.080 & 0.099 & 0.166 & 0.121 \\ \hline
\multirow{2}{*}{30} & All & 0.117 & 0.107 & 0.108 & 0.117 & 0.195 & 0.155 \\ \cline{2-8} 
                    & Remove 5\% outliers  & 0.079 & 0.090 & 0.074 & 0.095 & 0.157 & 0.115 \\ \hline
\end{tabular}
\label{tb:2clusters}
\end{table}

\begin{table}[ht]
\caption{Comparison of RMSE under three clusters scenario with different number of $Y^*$. ``Remove 5\% outliers'' means that calculates RMSE after removing the top 5\% of data (under bias measure). ``Remove 15\% outliers'' means that calculates RMSE after removing the top 15\% of data (under bias measure).}
\centering
\small
\begin{tabular}{|c|c|cccccc|}
\hline
\multirow{2}{*}{Number of $Y^*$} & \multirow{2}{*}{Data}                    & \multicolumn{6}{c|}{Root mean square error (RMSE)}                                                                                                                                \\ \cline{3-8} 
                                                &                                          & \multicolumn{1}{c|}{BCMAP} & \multicolumn{1}{c|}{MAP}   & \multicolumn{1}{c|}{rBCMAP} & \multicolumn{1}{c|}{rMAP}  & \multicolumn{1}{c|}{NPP}   & MEM   \\ \hline
\multirow{3}{*}{5}                              & All                                      & \multicolumn{1}{c|}{0.228} & \multicolumn{1}{c|}{0.214} & \multicolumn{1}{c|}{0.218}  & \multicolumn{1}{c|}{0.250} & \multicolumn{1}{c|}{0.438} & 0.287 \\ \cline{2-8} 
                                                & \multicolumn{1}{l|}{Remove 5\% outliers} & \multicolumn{1}{c|}{0.192} & \multicolumn{1}{c|}{0.181} & \multicolumn{1}{c|}{0.184}  & \multicolumn{1}{c|}{0.192} & \multicolumn{1}{c|}{0.349} & 0.206 \\ \cline{2-8} 
                                                & Remove 15\% outliers                     & \multicolumn{1}{c|}{0.142} & \multicolumn{1}{c|}{0.144} & \multicolumn{1}{c|}{0.136}  & \multicolumn{1}{c|}{0.151} & \multicolumn{1}{c|}{0.305} & 0.161 \\ \hline
\multirow{3}{*}{10}                             & All                                      & \multicolumn{1}{c|}{0.207} & \multicolumn{1}{c|}{0.191} & \multicolumn{1}{c|}{0.197}  & \multicolumn{1}{c|}{0.211} & \multicolumn{1}{c|}{0.368} & 0.233 \\ \cline{2-8} 
                                                & \multicolumn{1}{l|}{Remove 5\% outliers} & \multicolumn{1}{c|}{0.170} & \multicolumn{1}{c|}{0.161} & \multicolumn{1}{c|}{0.161}  & \multicolumn{1}{c|}{0.169} & \multicolumn{1}{c|}{0.319} & 0.173 \\ \cline{2-8} 
                                                & Remove 15\% outliers                     & \multicolumn{1}{c|}{0.118} & \multicolumn{1}{c|}{0.123} & \multicolumn{1}{c|}{0.111}  & \multicolumn{1}{c|}{0.128} & \multicolumn{1}{c|}{0.274} & 0.126 \\ \hline
\multirow{3}{*}{15}                             & All                                      & \multicolumn{1}{c|}{0.184} & \multicolumn{1}{c|}{0.170} & \multicolumn{1}{c|}{0.174}  & \multicolumn{1}{c|}{0.184} & \multicolumn{1}{c|}{0.309} & 0.204 \\ \cline{2-8} 
                                                & \multicolumn{1}{l|}{Remove 5\% outliers} & \multicolumn{1}{c|}{0.146} & \multicolumn{1}{c|}{0.141} & \multicolumn{1}{c|}{0.138}  & \multicolumn{1}{c|}{0.147} & \multicolumn{1}{c|}{0.269} & 0.160 \\ \cline{2-8} 
                                                & Remove 15\% outliers                     & \multicolumn{1}{c|}{0.096} & \multicolumn{1}{c|}{0.106} & \multicolumn{1}{c|}{0.094}  & \multicolumn{1}{c|}{0.110} & \multicolumn{1}{c|}{0.219} & 0.116 \\ \hline
\multirow{3}{*}{20}                             & All                                      & \multicolumn{1}{c|}{0.171} & \multicolumn{1}{c|}{0.159} & \multicolumn{1}{c|}{0.162}  & \multicolumn{1}{c|}{0.169} & \multicolumn{1}{c|}{0.274} & 0.184 \\ \cline{2-8} 
                                                & \multicolumn{1}{l|}{Remove 5\% outliers} & \multicolumn{1}{c|}{0.132} & \multicolumn{1}{c|}{0.127} & \multicolumn{1}{c|}{0.126}  & \multicolumn{1}{c|}{0.132} & \multicolumn{1}{c|}{0.235} & 0.143 \\ \cline{2-8} 
                                                & Remove 15\% outliers                     & \multicolumn{1}{c|}{0.084} & \multicolumn{1}{c|}{0.094} & \multicolumn{1}{c|}{0.082}  & \multicolumn{1}{c|}{0.098} & \multicolumn{1}{c|}{0.179} & 0.103 \\ \hline
\multirow{3}{*}{25}                             & All                                      & \multicolumn{1}{c|}{0.164} & \multicolumn{1}{c|}{0.150} & \multicolumn{1}{c|}{0.157}  & \multicolumn{1}{c|}{0.159} & \multicolumn{1}{c|}{0.254} & 0.175 \\ \cline{2-8} 
                                                & \multicolumn{1}{l|}{Remove 5\% outliers} & \multicolumn{1}{c|}{0.126} & \multicolumn{1}{c|}{0.120} & \multicolumn{1}{c|}{0.121}  & \multicolumn{1}{c|}{0.125} & \multicolumn{1}{c|}{0.215} & 0.132 \\ \cline{2-8} 
                                                & Remove 15\% outliers                     & \multicolumn{1}{c|}{0.079} & \multicolumn{1}{c|}{0.089} & \multicolumn{1}{c|}{0.078}  & \multicolumn{1}{c|}{0.090} & \multicolumn{1}{c|}{0.156} & 0.095 \\ \hline
\multirow{3}{*}{30}                             & All                                      & \multicolumn{1}{c|}{0.153} & \multicolumn{1}{c|}{0.138} & \multicolumn{1}{c|}{0.146}  & \multicolumn{1}{c|}{0.146} & \multicolumn{1}{c|}{0.229} & 0.162 \\ \cline{2-8} 
                                                & \multicolumn{1}{l|}{Remove 5\% outliers} & \multicolumn{1}{c|}{0.116} & \multicolumn{1}{c|}{0.111} & \multicolumn{1}{c|}{0.110}  & \multicolumn{1}{c|}{0.114} & \multicolumn{1}{c|}{0.191} & 0.122 \\ \cline{2-8} 
                                                & Remove 15\% outliers                     & \multicolumn{1}{c|}{0.072} & \multicolumn{1}{c|}{0.080} & \multicolumn{1}{c|}{0.070}  & \multicolumn{1}{c|}{0.083} & \multicolumn{1}{c|}{0.132} & 0.088 \\ \hline
\end{tabular}
\label{tb:3clusters}
\end{table}

\begin{table}[ht]
\small
\centering
\caption{The simulated external datasets. Homogeneous scenario: one cluster, heterogeneous scenarios: two clusters and three clusters.}
\begin{tabular}{|c|c|c|}
  \hline
Cluster & Number of Datasets & Sample Size in Each Dataset  \\
\hline 
1       & 10                 & 17, 71, 62, 72, 57, 31, 71, 94, 57, 15  
 \\ \hline 
2       & 12                 & 88, 50, 56, 99, 69, 25, 15, 81, 95, 95, 48, 40  
 \\ \hline
3       & 25                 & \begin{tabular}[c]{@{}l@{}}  39, 98, 25, 97, 63, 84, 57, 29, 76, 45, 61, 31, \\ 58, 51, 68, 93, 20, 64, 17, 55, 94, 75, 86, 55, 79 \end{tabular} 
 \\ \hline
\end{tabular}
\label{tb:simParaEst}
\end{table}

\begin{figure}[ht]
\centering
\includegraphics[width=\textwidth]{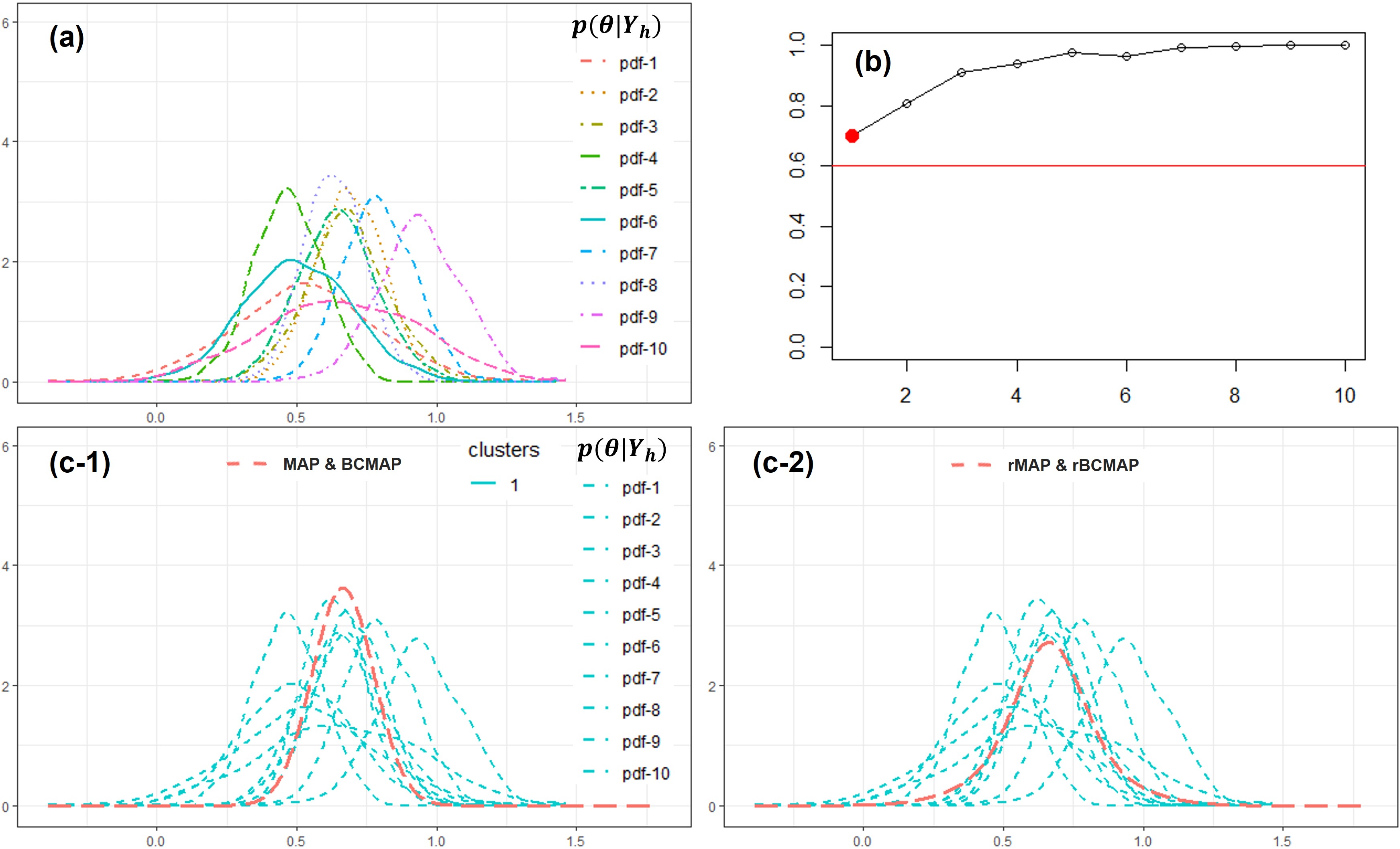}
\caption[bcp]{External data with one cluster. Panel (a): the posteriors with weakly informative prior. Panel (b): the optimal number of clusters. Panel (c-1): the synthesized MAP and BCMAP. Panel (c-2): the synthesized rMAP and rBCMAP. }
\label{fig:simOneCluster}
\end{figure}

\begin{figure}[ht]
\centering
\includegraphics[width=\textwidth]{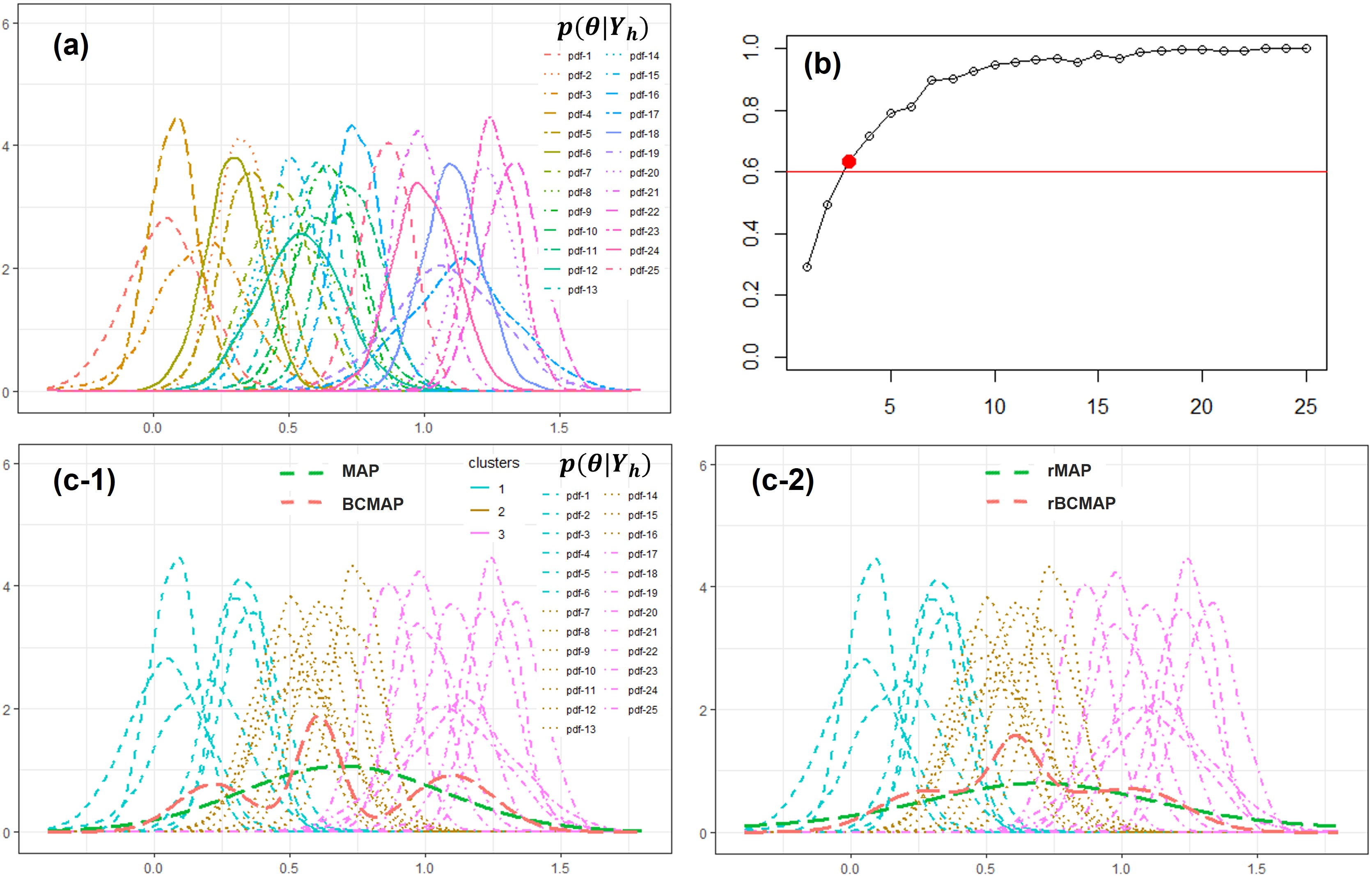}
\caption[bcp]{External data with three clusters. Panel (a): the posteriors with weakly informative prior. Panel (b): the optimal number of clusters. Panel (c-1): the synthesized MAP and BCMAP. Panel (c-2): the synthesized rMAP and rBCMAP. }
\label{fig:simThreeClusters}
\end{figure}

\begin{table}[ht]
\small
\caption{Studies on acupoint P6 stimulation for preventing nausea. The variables are: \textbf{id:} trial id number. \textbf{study:} first author of the study. \textbf{year:} study year. \textbf{ai:} number of patients experiencing nausea in treatment (wrist acupuncture point P6 treatment) group. \textbf{n1i:} total number of patients in treatment group.}
\centering
\begin{tabular}{lllll}
\hline
\textbf{id} & \textbf{study} & \textbf{year} & \textbf{ai ($r_h$)} & \textbf{n1i ($N_h$)}  \\ \hline
1           & Agarwal        & \textit{2000} & \textbf{18} & \textbf{100}         \\
2           & Agarwal        & \textit{2002} & \textbf{5}  & \textbf{50}       \\
3           & Alkaissi       & \textit{1999} & \textbf{9}  & \textbf{20}          \\
4           & Alkaissi       & \textit{2002} & \textbf{32} & \textbf{135}         \\
5           & Allen          & \textit{1994} & \textbf{9}  & \textbf{23}           \\
6           & Andrzejowski   & \textit{1996} & \textbf{11} & \textbf{18}            \\
7           & Duggal         & \textit{1998} & \textbf{69} & \textbf{122}     \\
8           & Dundee         & \textit{1986} & \textbf{3}  & \textbf{25}        \\
9           & Ferrera-Love   & \textit{1996} & \textbf{1}  & \textbf{30}       \\
10          & Gieron         & \textit{1993} & \textbf{11} & \textbf{30}         \\
11          & Harmon         & \textit{1999} & \textbf{7}  & \textbf{44}       \\
12          & Harmon         & \textit{2000} & \textbf{4}  & \textbf{47}        \\
13          & Ho             & \textit{1996} & \textbf{1}  & \textbf{30}         \\
14          & Rusy           & \textit{2002} & \textbf{24} & \textbf{40}          \\
15          & Wang           & \textit{2002} & \textbf{16} & \textbf{50}         \\
16          & Zarate         & \textit{2001} & \textbf{28} & \textbf{110}      \\ \hline
\end{tabular}
\label{tb:real}
\end{table}

\begin{figure}[t!]
\centering
\includegraphics[width=\textwidth]{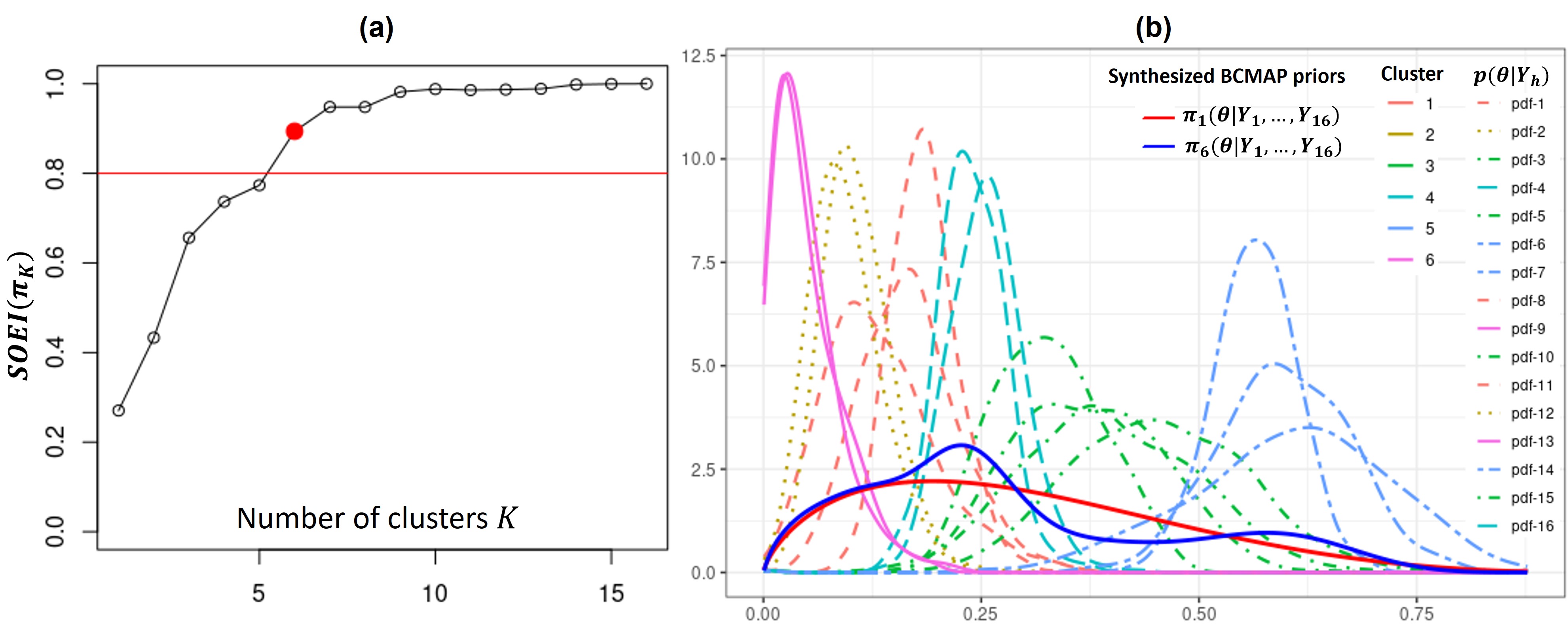}
\caption[bcp]{Bayesian clustering prior constructed on the wrist acupuncture point P6 treatment data from 16 studies. Panel (a) shows the selection of optimal number of clusters $K^*=6$. Panel (b) presents the clustering results and corresponding parameterized BCMAP prior $\pi_6(p|Y_1,\ldots,Y_{16})$. As a comparison, the parameterized MAP prior $\pi_1(p|Y_1,\ldots,Y_{16})$ is also shown in Panel (b). }
\label{fig:realData2}
\end{figure}

\end{appendix}
\FloatBarrier


\section{Funding}
J. Jack Lee's research was supported in part by the grants P30CA016672, P01CA285249, P50CA221703, U24CA224285, and U24CA224020 from the National Cancer Institute.


\newpage

\bibliographystyle{chicago}

\bibliography{BCPOI}

\begin{thebibliography}{}

\bibitem[\protect\citeauthoryear{Carlin and Louis}{Carlin and
  Louis}{2000}]{Carlin2000EmpiricalBP}
Carlin, B.~P. and T.~A. Louis (2000).
\newblock Empirical bayes: Past, present and future.
\newblock {\em Journal of the American Statistical Association\/}~{\em 95},
  1286 -- 1289.

\bibitem[\protect\citeauthoryear{Chen, Ibrahim, and Shao}{Chen
  et~al.}{2000}]{CHEN2000PP}
Chen, M.-H., J.~G. Ibrahim, and Q.-M. Shao (2000).
\newblock Power prior distributions for generalized linear models.
\newblock {\em Journal of Statistical Planning and Inference\/}~{\em 84\/}(1),
  121--137.

\bibitem[\protect\citeauthoryear{Chen and Lee}{Chen and Lee}{2020}]{Chen2020}
Chen, N. and J.~J. Lee (2020).
\newblock Bayesian cluster hierarchical model for subgroup borrowing in the
  design and analysis of basket trials with binary endpoints.
\newblock {\em Statistical Methods in Medical Research\/}~{\em 29\/}(9),
  2717--2732.
\newblock PMID: 32178585.

\bibitem[\protect\citeauthoryear{Food and Administration}{Food and
  Administration}{2017}]{FDA2017}
Food, U. and D.~Administration (2017).
\newblock Use of real-world evidence to support regulatory decision-making for
  medical devices: guidance for industry and food and drug administration
  staff.
\newblock {\em Silver Spring\/}.

\bibitem[\protect\citeauthoryear{Gershman and Blei}{Gershman and
  Blei}{2012}]{GERSHMAN20121}
Gershman, S.~J. and D.~M. Blei (2012).
\newblock A tutorial on bayesian nonparametric models.
\newblock {\em Journal of Mathematical Psychology\/}~{\em 56\/}(1), 1--12.

\bibitem[\protect\citeauthoryear{Hobbs, Sargent, and Carlin}{Hobbs
  et~al.}{2012}]{Hobbs2012}
Hobbs, B., D.~Sargent, and B.~Carlin (2012).
\newblock Commensurate priors for incorporating historical information in
  clinical trials using general and generalized linear models.
\newblock {\em Bayesian Analysis\/}~{\em 7\/}(3), 639--674.

\bibitem[\protect\citeauthoryear{Hobbs, Carlin, Mandrekar, and Sargent}{Hobbs
  et~al.}{2011}]{Hobbs2011}
Hobbs, B.~P., B.~P. Carlin, S.~J. Mandrekar, and D.~J. Sargent (2011).
\newblock Hierarchical commensurate and power prior models for adaptive
  incorporation of historical information in clinical trials.
\newblock {\em Biometrics\/}~{\em 67\/}(3), 1047--1056.

\bibitem[\protect\citeauthoryear{Jiang, Nie, and Yuan}{Jiang
  et~al.}{2023}]{Jiang2023}
Jiang, L., L.~Nie, and Y.~Yuan (2023).
\newblock Elastic priors to dynamically borrow information from historical data
  in clinical trials.
\newblock {\em Biometrics\/}~{\em 79\/}(1), 49--60.

\bibitem[\protect\citeauthoryear{Joseph G~Ibrahim and Sinha}{Joseph G~Ibrahim
  and Sinha}{2003}]{Ibrahim2003}
Joseph G~Ibrahim, M.-H.~C. and D.~Sinha (2003).
\newblock On optimality properties of the power prior.
\newblock {\em Journal of the American Statistical Association\/}~{\em
  98\/}(461), 204--213.

\bibitem[\protect\citeauthoryear{Kaizer, Koopmeiners, and Hobbs}{Kaizer
  et~al.}{2017}]{Kaizer2017}
Kaizer, A.~M., J.~S. Koopmeiners, and B.~P. Hobbs (2017, 07).
\newblock {Bayesian hierarchical modeling based on multisource
  exchangeability}.
\newblock {\em Biostatistics\/}~{\em 19\/}(2), 169--184.

\bibitem[\protect\citeauthoryear{Lee and Done}{Lee and Done}{2004}]{Lee04}
Lee, A. and M.~Done (2004).
\newblock Stimulation of the wrist acupuncture point p6 for preventing
  postoperative nausea and vomiting.
\newblock {\em Cochrane Database of Systematic Reviews\/}~(3).

\bibitem[\protect\citeauthoryear{Lu and Lee}{Lu and
  Lee}{2023}]{lu2023overlapping}
Lu, X. and J.~J. Lee (2023).
\newblock Overlapping indices for dynamic information borrowing in bayesian
  hierarchical modeling.
\newblock \url{https://arxiv.org/abs/2305.17515v1}.

\bibitem[\protect\citeauthoryear{Neuenschwander, Capkun-Niggli, Branson, and
  Spiegelhalter}{Neuenschwander et~al.}{2010}]{Neuenschwander2010}
Neuenschwander, B., G.~Capkun-Niggli, M.~Branson, and D.~J. Spiegelhalter
  (2010).
\newblock Summarizing historical information on controls in clinical trials.
\newblock {\em Clinical Trials\/}~{\em 7\/}(1), 5--18.

\bibitem[\protect\citeauthoryear{Schmid and Schmidt}{Schmid and
  Schmidt}{2006}]{Schmid2006}
Schmid, F. and A.~Schmidt (2006).
\newblock Nonparametric estimation of the coefficient of overlapping-theory and
  empirical application.
\newblock {\em Comput. Stat. Data Anal.\/}~{\em 50\/}(6), 1583–1596.

\bibitem[\protect\citeauthoryear{Schmidli, Gsteiger, Roychoudhury, O'Hagan,
  Spiegelhalter, and Neuenschwander}{Schmidli et~al.}{2014}]{Schmidli2014}
Schmidli, H., S.~Gsteiger, S.~Roychoudhury, A.~O'Hagan, D.~Spiegelhalter, and
  B.~Neuenschwander (2014).
\newblock {Robust Meta-Analytic-Predictive Priors in Clinical Trials with
  Historical Control Information}.
\newblock {\em Biometrics\/}~{\em 70\/}(4), 1023--1032.

\bibitem[\protect\citeauthoryear{Spiegelhalter, Abrams, and
  Myles}{Spiegelhalter et~al.}{2004}]{Spiegelhalter2004}
Spiegelhalter, D. J. D.~J., K.~R. K.~R. Abrams, and J.~P. Myles (2004).
\newblock {\em Bayesian approaches to clinical trials and health-care
  evaluation}.
\newblock Statistics in practice. Chichester: John Wiley \& Sons.

\bibitem[\protect\citeauthoryear{Weber, Li, Seaman~III, Kakizume, and
  Schmidli}{Weber et~al.}{2021}]{RBEST2021}
Weber, S., Y.~Li, J.~W. Seaman~III, T.~Kakizume, and H.~Schmidli (2021).
\newblock Applying meta-analytic-predictive priors with the r bayesian evidence
  synthesis tools.
\newblock {\em Journal of Statistical Software\/}~{\em 100\/}(19), 1–32.

\bibitem[\protect\citeauthoryear{Weitzman}{Weitzman}{1970}]{Weitzman1970}
Weitzman, M.~S. (1970).
\newblock Measures of overlap of income distributions of white and negro
  families in the united states.
\newblock {\em Technical Report 22, US Department of Commerce\/}.

\bibitem[\protect\citeauthoryear{White, Noble, Senior, Hamilton, and
  Viechtbauer}{White et~al.}{2022}]{metadat2022}
White, T., D.~Noble, A.~Senior, W.~K. Hamilton, and W.~Viechtbauer (2022).
\newblock {\em metadat: Meta-Analysis Datasets}.
\newblock R package version 1.2-0.

\end{thebibliography}
\end{document}